\documentclass[conference]{IEEEtran}

\usepackage{graphicx}
\usepackage{enumerate}
\usepackage{amsmath}
\usepackage{amsthm}
\usepackage{amssymb}
\usepackage{amscd}
\newtheorem{definition}{Definition}
\newtheorem{lem}{Lemma}
\newtheorem{thm}{Theorem}
\let\oldhat\hat
\renewcommand{\vec}[1]{\mathbf{#1}}
\renewcommand{\hat}[1]{\oldhat{\mathbf{#1}}}

\begin{document}

\title{Power Allocation over Two Identical Gilbert-Elliott Channels }
\author{\IEEEauthorblockN{Junhua Tang}
\IEEEauthorblockA{School of Electronic Information\\ and Electrical Engineering\\
Shanghai Jiao Tong University, China\\
Email: junhuatang@sjtu.edu.cn}
\and
\IEEEauthorblockN{Parisa Mansourifard}
\IEEEauthorblockA{Ming Hsieh Department \\
of Electrical Engineering\\
Viterbi School of Engineering \\
University of Southern California\\
Email: parisama@usc.edu}
\and
\IEEEauthorblockN{Bhaskar Krishnamachari}
\IEEEauthorblockA{Ming Hsieh Department \\
 of Electrical Engineering\\
Viterbi School of Engineering \\
University of Southern California\\
Email: bkrishna@usc.edu}
}
\maketitle

\begin{abstract}
We study the problem of power allocation over two identical Gilbert-Elliot communication channels. Our goal is to maximize the expected discounted number of bits transmitted over an infinite time horizon. This is achieved by choosing among three possible strategies: (1) betting on channel 1 by allocating all the power to this channel, which results in high data rate if channel 1 happens to be in good state, and zero bits transmitted if channel 1 is in bad state (even if channel 2 is in good state) (2) betting on channel 2 by allocating all the power to the second channel, and (3) a balanced strategy whereby each channel is allocated half the total power, with the effect that each channel can transmit a low data rate if it is in good state. We assume that each channel's state is only revealed upon transmission of data on that channel. We model this problem as a partially observable Markov decision processes (MDP), and derive key threshold properties of the optimal policy. Further, we show that by formulating and solving a relevant linear program the thresholds can be determined numerically when system parameters are known.
\end{abstract}


\section{Introduction}{\label{sec:intro}}

Adaptive power control is an important technique to select the transmission power of a wireless system according to channel condition to achieve better network performance in terms of higher data rate or spectrum efficiency \cite{yoo2006}, \cite{weiyu2004}. While there has been some recent work on power allocation over stochastic channels~\cite{zaidi05, wang10, gai12}, the problem of optimal adaptive power allocation across multiple stochastic channels with memory is challenging and poorly understood. In this paper, we analyze a simple but fundamental problem.  We consider a wireless system operating on two stochastically identical independent parallel transmission channels, each modeled as a slotted Gilber-Elliott channel (i.e. described by two-state Markov chains, with a bad state ``0'' and a good state ``1''). Our objective is to allocate the limited power budget to the two channels dynamically so as to maximize the expected discounted number of bits transmitted over time. Since the channel state is unknown when power allocation decision is made, this problem is more challenging than it looks like.

Recently, several works have explored different sequential decision-making problems involving Gilbert-Elliott channels~\cite{zhao2008}, \cite{Ahmad2009}, \cite{laourine2010}, \cite{yanting2012}, \cite{nayyar2011}. In~\cite{zhao2008}, \cite{Ahmad2009}, the authors consider selecting one channel to sense/access at each time among several identical channels, formulate it as a restless multi-armed problem, and show that a simple myopic policy is optimal whenever the channels are positively correlated over time.  In \cite{laourine2010}, the authors study the problem of dynamically choosing one of three transmitting schemes for a single Gilbert-Elliott channel in an attempt to maximize the expected discounted number of bits transmitted. And in \cite{yanting2012}, the authors study the problem of choosing a transmitting strategy from two choices emphasizing the case when the channel transition probabilities are unknown. While similar in spirit to these two studies, our work addresses a more challenging setting involving two independent channels. A more related two-channel problem is studied in \cite{nayyar2011}, which characterizes the optimal policy to opportunistically access two non-identical Gilber-Elliott channels (generalizing the prior work on sensing policies for identical channels~\cite{zhao2008}, \cite{Ahmad2009}). While we address only identical channels in this work, the strategy space explored here is richer because in our formulation of power allocation, it is possible to use both channels simultaneously whilst in \cite{zhao2008}, \cite{Ahmad2009}, \cite{nayyar2011} only one channel is accessed in each time slot.


In this paper, we formulate our power allocation problem as a partially observable Markov decision process (POMDP). We then treat the POMDP as a continuous state MDP and develop the structure of the optimal policy (decision). Our main contributions are the following: (1) we formulate the problem of dynamic power allocation over parallel Markovian channels, (2) using the MDP theory, we theoretically prove key threshold properties of the optimal policy for this particular problem, (3) through simulation based on linear programming, we demonstrate the existence of the 0-threshold and 2-threshold structures of the optimal policy, and (4) we demonstrate how to numerically compute the thresholds and construct the optimal policy when system parameters are known.

\section{Problem Formulation}

\subsection*{A. Channel model and assumptions}

We consider a wireless communication system operating on two parallel channels. Each channel is described by a slotted Gilbert-Elliott model which is a one dimensional Markov chain $G_{i,t} ( i \in \{1,2\}, t \in \{1,2,...,\infty\})$ with two states: a good state denoted by 1 and a bad state denoted by 0 ($i$ is the channel number and $t$ is the time slot). The channel transition probabilities are given by $Pr[G_{i,t}=1|G_{i,t-1}=1]=\lambda_1, i \in \{1,2\} $ and  $Pr[G_{i,t}=1|G_{i,t-1}=0]=\lambda_0, i \in \{1,2\} $. We assume the two channels are identical and independent of each other, and channel transitions occur at the beginning of each time slot. We also assume that $\lambda_0 \leq \lambda_1$, which is the positive correlation assumption commonly used in the literature.

The system has a total transmission power of $P$. At the beginning of time slot $t$, the system allocates transmission power $P_1(t)$ to channel 1 and $P_2(t)$ to channel 2, where $P_1(t)+P_2(t) = P$. We assume the channel state is not directly observable at the beginning of each time slot. That is, the system needs to allocate the transmission power to the two parallel channels without knowing the channel states. If channel $i (i \in \{1,2\})$ is used at time slot $t$ by allocating transmission power $P_i(t)$ on it, the channel state of the elapsed slot is revealed at the end of the time slot through channel feedback. But if a channel is not used, that is, if transmission power is 0 on that channel, the channel state of the elapsed slot remains unknown at the end of that slot.

\subsection*{B. Power allocation strategies}

To simplify the problem, we assume the system may allocate one of the following three power levels to a channel: $0, P/2$, or $P$. That is, based on the belief in the channel state of channel $i$ for the current time slot $t$, the system may decide to give up the channel ($P_i(t)=0$), use it moderately ($P_i(t)=P/2$) or use it fully($P_i(t)=P$). Since the channel state is not directly observable when the power allocation is done, the following circumstances may occur. If a channel is in bad state, no data is transmitted at all no matter what the allocated power is. If a channel is in good state, and  power $P/2$ is allocated to it, it can transmit $R_l$ bits of data successfully during that slot. If a channel is in good condition and  power $P$ is allocated to it, it can transmit $R_h$ bits of data successfully during that slot. We assume $R_l<R_h<2R_l$.

We define three power allocation strategies(actions): balanced, betting on channel 1, and betting on channel 2. Each strategy is explained in detail as follows.

\emph{Balanced:} For this action (denoted by $B_b$), the system allocates the transmission power evenly on both channels, that is, $P_1(t)=P_2(t)=P/2$, for time slot $t$. This corresponds to the situation when the system cannot determine which of the channels is more likely to be in good state, so it decides to ``play safe" by using both of the channels.

\emph{Betting on channel 1:} For this action (denoted by $B_1$), the system decides to ``gamble" and allocate all the transmission power to channel 1. That is, $P_1(t)=P, P_2(t)=0$ for time slot $t$.
This corresponds to the situation when the system believes that channel 1 is in a good state and channel 2 is in a bad state.

\emph{Betting on channel 2:} For this action (denoted by $B_2$), the system put all the transmission power in channel 2, that is, $P_2(t)=P, P_1(t)=0$ for time slot $t$.

Note that for strategies $B_1$ and $B_2$, if a channel is not used, the system (transmitter) will not acquire any knowledge about the state of that channel during the elapsed slot.
\subsection*{C. POMDP formulation}

At the beginning of a time slot, the system is confronted with a choice among three actions. It must judiciously select actions so as to maximize the total expected discounted number of bits transmitted over an infinite time span. Because the state of the channels is not directly observable, the problem in hand is a Partially Observable Markov Decision Process (POMDP). In \cite{smallwood1973}, it is shown that a sufficient statistic for determining the optimal policy is the conditional probability that the channel is in the good state at the beginning of the current slot given the past history (henceforth called belief) \cite{laourine2010}. Denote the belief of the system by a two dimension vector $\vec{x}_t$=$(x_{1,t},x_{2,t})$, where $x_{1,t}= \Pr [G_{1,t}=1|\hbar _t]$, $x_{2,t}= \Pr [G_{2,t}=1|\hbar _t]$, where $\hbar_t$ is all the history of actions and observations at the current slot $t$. By using this belief as the decision variable, the POMDP problem is converted into an MDP with the uncountable state space $([0,1],[0,1])$ \cite{laourine2010}.

Define a policy $\pi$ as a rule that dictates the action to choose, \emph{i.e.}, a map from the belief at a particular time to an action in the action space. Let $V^\pi(\vec{p})$ be the expected discounted reward with initial belief $\vec{p}=(p_1,p_2)$, that is, $x_{1,0}=\Pr [G_{1,0}=1|\hbar_0]=p_1$, $x_{2,0}=\Pr [G_{2,0}=1|\hbar_0]=p_2$, where the superscript $\pi$ denotes the policy being followed. Define $\beta (\in [0,1))$ as the discount factor, the expected discounted reward has the following expression
\begin{equation}
V^\pi(\vec{p})=E_\pi [ \sum_{t=0}^\infty \beta^t g_{a_t}(\vec{x}_t)|\vec{x}_0=\vec{p}],
\label{eq:vbpip}
\end{equation}
where $E_\pi$ represents the expectation given that the policy $\pi$ is employed, $t$ is the time slot index, $a_t$ is the action chosen at time $t$, $a_t \in \{B_b, B_1, B_2\}$. The term $g_{a_t}(\vec{x}_t)$ denotes the expected reward acquired when the belief is $\vec{x_t}$ and the action $a_t$ is chosen:
\[ g_{a_t}(\vec{x}_t)= \left \{
\begin{array}{rl}
x_{1,t}R_l+x_{2,t}R_l, & \quad \mbox{if} \quad  a_t = B_b\\
x_{1,t}R_h, & \quad \mbox{if} \quad   a_t = B_1\\
x_{2,t}R_h, & \quad \mbox{if} \quad   a_t = B_2\\
\end{array} .\right. \]
\begin{equation}
\end{equation}
Now we define the value function $V(\vec{p})$ as
\begin{equation}
V(\vec{p})=\max_\pi V^\pi (\vec{p}), \quad \mbox{for all} \quad \vec{p} \in ([0,1],[0,1]).
\label{eq:vbp}
\end{equation}
A policy is said to be stationary if it is a function mapping the state space $([0,1],[0,1])$ into the action space $\{B_b, B_1, B_2\}$. Ross proved in \cite{ross1970} (Th.6.3) that there exists a stationary policy $\pi^*$ such that $V(\vec{p})=V^{\pi^*}(\vec{p})$. The value function $V(\vec{p})$ satisfies the Bellman equation
\begin{equation}
V(\vec{p})=\max_{a \in \{B_b, B_1, B_2\}} \{ V_a(\vec{p})\},
\label{eq:bellman}
\end{equation}
where $V_a(\vec{p})$ is the value acquired by taking action $a$ when the initial belief is $\vec{p}$. $V_a(\vec{p})$ is given by
\begin{equation}
V_a(\vec{p})= g_a(\vec{p})+\beta E^\vec{y} [V(\vec{y})|\vec{x}_0=\vec{p}, a_0=a],
\label{eq:vbap}
\end{equation}
where $\vec{y}$ denotes the next belief when the action $a$ is chosen and the initial belief is $\vec{p}$. The term $V_a(\vec{p})$ is explained next for the three possible actions.

\emph{a) Balanced (action $B_b$):} If this action is taken, and the current belief is $\vec{p}=(p_1,p_2)$, the immediate reward is $p_1R_l+p_2R_l$. Since both channels are used, the channel quality of both channels during the current slot is then revealed to the transmitter. With probability $p_1$ the first channel will be in good state and hence the belief of channel 1 at the beginning of the next slot will be $\lambda_1$. Likewise, with probability $1-p_1$ channel 1 will turn out to be in bad state and hence the updated belief of channel 1 for the next slot is $\lambda_0$. Since channel 2 and channel 1 are identical, channel 2 has similar belief update. Consequently if action $B_b$ is taken, the value function evolves as
\begin{eqnarray}
\label{eq:vbbb}
& & V_{B_b}(p_1,p_2) \nonumber \\
&=&p_1R_l+p_2R_l+\beta[(1-p_1)(1-p_2)V(\lambda_0,\lambda_0)\nonumber \\
&+&p_1(1-p_2)V(\lambda_1,\lambda_0)+(1-p_1)p_2V(\lambda_0,\lambda_1) \nonumber\\
&+&p_1p_2V(\lambda_1,\lambda_1)].
\end{eqnarray}
\emph{b) Betting on channel 1( action $B_1$):} If this action is taken, and the current belief is $\vec{p}=(p_1,p_2)$, the immediate reward is $p_1R_h$. But since channel 2 is not used, its channel state remains unknown. Hence if the belief of channel 2 during the elapsed time slot is $p_2$, its belief at the beginning of the next time slot is given by
\begin{equation}
T(p_2)= p_2\lambda_1+(1-p_2)\lambda_0=\alpha p_2 + \lambda_0,
\end{equation}
where $\alpha=\lambda_1-\lambda_0$. Consequently, if this action is taken, the value function evolves as
\begin{eqnarray}
\label{eq:vbb1}
& &V_{B_1}(p_1,p_2)= p_1R_h+ \nonumber \\
& &\beta[(1-p_1)V(\lambda_0,T(p_2))+p_1V(\lambda_1,T(p_2))].
\end{eqnarray}
\emph{c) Betting on channel 2(action $B_2$):} Similar to action $B_1$, if action $B_2$ is taken, the value function evolves as
\begin{eqnarray}
\label{eq:vbb2}
& &V_{B_2}(p_1,p_2)= p_2R_h+ \nonumber \\
& &\beta[(1-p_2)V(T(p_1),\lambda_0)+p_2V(T(p_1),\lambda_1)],
\end{eqnarray}
where
\begin{equation}
T(p_1)=p_1\lambda_1+(1-p_1)\lambda_0=\alpha p_1 + \lambda_0.
\end{equation}
Finally the Bellman equation for our power allocation problem reads as follows
\begin{equation}
V(\vec{p})=\max \{V_{B_b}(\vec{p}),V_{B_1}(\vec{p}),V_{B_2}(\vec{p})\}.
\end{equation}


\section{Structure of the Optimal Policy}
From the above discussion we understand that an optimal policy exists for our power allocation problem. In this section, we try to derive the optimal policy by first looking at the features of its structure.
\subsection*{A: Properties of value function}
\begin{lem} \label{lem:linear}
$V_{B_i}(p_1,p_2), i \in \{1,2,b\}$ is affine with respect to $p_1$ and $p_2$ and the following equalities hold:
\begin{eqnarray}\label{eq:linear}
V_{B_i}(cp+(1-c)p',p_2)=cV_{B_i}(p,p_2)+(1-c)V_{B_i}(p',p_2), \nonumber \\
V_{B_i}(p_1,cp+(1-c)p')=cV_{B_i}(p_1,p)+(1-c)V_{B_i}(p_1,p'),
\end{eqnarray}
where $0 \leq c \leq 1$ is a constant; and we say $f(x)$ is affine with respect to $x$ if $f(x)=a+cx$, with constant $a$ and $c$.
\end{lem}
\begin{proof}

It is clear that $V_{B_b}$ is affine in $p_1$ and $p_2$ from (\ref{eq:vbbb}). It is also obvious that $V_{B_1}$ is affine in $p_1$ and $V_{B_2}$ is affine in $p_2$ from (\ref{eq:vbb1}) and (\ref{eq:vbb2}). Next we will prove that $V_{B_1}$ is affine in $p_2$ and $V_{B_2}$ is affine in $p_1$, which will make the proof complete.

Now we prove that $V_{B_1}$ is affine in $p_2$. We will first show that the second term on the right side of (\ref{eq:vbb1}) is affine in $p_2$, the third term can then be shown to be affine in $p_2$ in a similar manner, thus the summation of the three terms in (\ref{eq:vbb1}) is affine in $p_2$.

Now let us look at the second term on the right side of (\ref{eq:vbb1}), the main part $V(\lambda_0, T(p_2))$ is one of the following three forms: $V_{B_b}(\lambda_0, T(p_2))$, $V_{B_2}(\lambda_0, T(p_2))$, or $V_{B_1}(\lambda_0, T(p_2))$. The first form $V_{B_b}(\lambda_0, T(p_2))$ is affine in $p_2$ because $V_{B_b}(\lambda_0, T(p_2))$ is affine in $T(p_2)$ and $T(p_2) = \alpha p_2 + \lambda_0$ is affine in $p_2$. Similarly, the second form  $V_{B_2}(\lambda_0, T(p_2))$ is affine in $T(p_2)$ thus affine in $p_2$. The third form $V_{B_1}(\lambda_0, T(p_2))$ is written as:
\begin{eqnarray}\label{eq:linearproof1}
&&V_{B_1}(\lambda_0, T(p_2)) \nonumber \\
&&= \lambda_0 R_h + \beta \lambda_0 V(\lambda_1,T^2(p_2)) + \beta (1-\lambda_0)V(\lambda_0, T^2(p_2)), \nonumber\\
\end{eqnarray}
where $T^n(p)=T(T^{n-1}(p))=\frac{\lambda_0}{1-\alpha}(1-\alpha^n)+\alpha^n p$.
Since $T^n(p_2)$ is affine in $p_2$, (\ref{eq:linearproof1}) is affine in $p_2$ as soon as $V(\lambda_1,T^n(p_2))$ takes the form of $V_{B_b}(\lambda_1,T^n(p_2))$ or $V_{B_2}(\lambda_1,T^n(p_2))$, and $V(\lambda_0, T^2(p_2))$ takes the form of $V_{B_b}(\lambda_0,T^n(p_2))$ or $V_{B_2}(\lambda_0,T^n(p_2))$, $n=2,3,4,...$, which is affine in $p_2$. If $V(\lambda_1,T^n(p_2))$ continues to take the form $V_{B_1}(\lambda_1,T^n(p_2))$ till $n$ goes to infinity, $V(\lambda_1,T^n(p_2))$ will eventually become a constant $V(\lambda_1,\frac{\lambda_0}{1-\alpha})$ because $T^n(p_2) \rightarrow \frac{\lambda_0}{1-\alpha}, n \rightarrow \infty$. Which means a special case of affine linearity in $p_2$. With this we show that the third form $V_{B_1}(\lambda_0, T(p_2))$ is affine in $p_2$. Therefore we have shown that $V(\lambda_0, T(p_2))$ is affine in $p_2$, thus the second term on the right side of (\ref{eq:vbb1}) is affine in $p_2$.

Similarly we can show that the third term on the right side of (\ref{eq:vbb1}) is affine in $p_2$, thus $V_{B_1}(p_1,p_2)$ is affine in $p_2$.

The affine linearity of $V_{B_2}(p_1,p_2)$ in $p_1$ can be proved using the same technique and the detail is omitted due to space limit.
\end{proof}
\begin{lem} $V_{B_i}(p_1,p_2), i \in \{1,2,b\}$ is convex in $p_1$ and $p_2$. 
\end{lem}
\begin{proof}
The convexity of $V_{B_i}, i \in \{1,2,b\}$ in $p_1$ and $p_2$ follows from its affine linearity in Lemma 1.
\end{proof}
\begin{lem}
\label{lem:vbsem}
$V(p_1,p_2)=V(p_2,p_1)$, that is, $V(p_1,p_2)$ is symmetric with respect to the line $p_1=p_2$ in the belief space.
\end{lem}
\begin{proof}
Define $V^{n}(p_1,p_2)$ as the optimal value when the decision horizon spans only $n$ stages. Then we have
\begin{eqnarray}
& & V^{1}(p_1,p_2) \nonumber \\
&=&\max\{V_{B_b}^{1}(p_1,p_2), V_{B_1}^{1}(p_1,p_2),V_{B_2}^{1}(p_1,p_2)\} \nonumber \\
&=&\max\{p_1R_l+p_2R_l,p_1R_h,p_2R_h \}.
\end{eqnarray}
\begin{eqnarray}
& & V^{1}(p_2,p_1) \nonumber \\
&=&\max\{V_{B_b}^{1}(p_2,p_1), V_{B_1}^{1}(p_2,p_1),V_{B_2}^{1}(p_2,p_1)\} \nonumber \\
&=&\max\{p_2R_l+p_1R_l,p_2R_h,p_1R_h \}.
\end{eqnarray}
It is easy to see that
\begin{equation}
V^1(p_1,p_2)=V^1(p_2,p_1).
\end{equation}
Assume $V^k(x_1,x_2)=V^k(x_2,x_1), k \geq 1$, next we will prove that $V^{k+1}(p_1,p_2)=V^{k+1}(p_2,p_1)$.
\begin{eqnarray}
& & V_{B_b}^{k+1}(p_1,p_2) \nonumber \\
&=&p_1R_l+p_2R_h+\beta[(1-p_1)(1-p_2)V^k(\lambda_0,\lambda_0) \nonumber \\
&+& p_1(1-p_2)V^k(\lambda_1,\lambda_0)+(1-p_1)p_2V^k(\lambda_0,\lambda_1) \nonumber \\
 &+& p_1p_2V^k(\lambda_1,\lambda_1)]
\end{eqnarray}
\begin{eqnarray}
& &V_{B_1}^{k+1}(p_1,p_2)= p_1R_h+ \nonumber \\
& &\beta[(1-p_1)V^k(\lambda_0,T(p_2))+p_1V^k(\lambda_1,T(p_2))].
\end{eqnarray}
\begin{eqnarray}
& &V_{B_2}^{k+1}(p_1,p_2)=p_2R_h+ \nonumber \\
& &\beta[(1-p_2)V^k(T(p_1),\lambda_0)+p_2V^k(T(p_1),\lambda_1)].
\end{eqnarray}
Using the assumption that $V^k(x_1,x_2)=V^k(x_2,x_1)$, it is easy to see that
\begin{eqnarray}
& & V_{B_1}^{k+1}(p_2,p_1) \nonumber \\
&=&p_2R_h+\beta[(1-p_2)V^k(\lambda_0,T(p_1))+p_2V^k(\lambda_1,T(p_1))] \nonumber \\
&=&p_2R_h+\beta[(1-p_2)V^k(T(p_1),\lambda_0)+p_2V^k(T(p_1),\lambda_1)] \nonumber \\
&=&V_{B_2}^{k+1}(p_1,p_2)
\end{eqnarray}
Similarly, we have $V_{B_2}^{k+1}(p_2,p_1)=V_{B_1}^{k+1}(p_1,p_2)$, and $V_{B_b}^{k+1}(p_2,p_1)=V_{B_b}^{k+1}(p_1,p_2)$, thus
\begin{eqnarray}
& & V^{k+1}(p_1,p_2) \nonumber \\
&=&\max\{V_{B_b}^{k+1}(p_1,p_2),V_{B_1}^{k+1}(p_1,p_2),V_{B_2}^{k+1}(p_1,p_2)\} \nonumber \\
&=& \max\{V_{B_b}^{k+1}(p_2,p_1),V_{B_2}^{k+1}(p_2,p_1),V_{B_1}^{k+1}(p_2,p_1)\} \nonumber \\
&=& V^{k+1}(p_2,p_1).
\end{eqnarray}
From the theory of MDPs, we know that $V^n(p_1,p_2) \rightarrow V(p_1,p_2)$ as $n \rightarrow \infty$.
Hence we have $V(p_1,p_2)=V(p_2,p_1)$, for any $(p_1,p_2)$ in the belief space.
\end{proof}

\subsection*{B: Properties of the decision regions of policy $\pi*$}

We use $\Phi_a$ to denote the set of beliefs for which it is optimal to take the action $a$. That is,
\begin{eqnarray}
\Phi_a= \{ (p_1,p_2) \in ([0,1],[0,1]), V(p_1,p_2)=V_a(p_1,p_2)\}, \nonumber\\
a \in \{ B_b, B_1, B_2 \}.
\end{eqnarray}
\begin{definition}
\label{def:contiguous}
$\Phi_a$ is said to be contiguous along $p_1$ dimension if we have $(x_1, p_2) \in \Phi_a$ and $(x_2, p_2) \in \Phi_a$, then $\forall x \in [x_1,x_2]$, we have $(x,p_2) \in \Phi_a$. Similarly, we say $\Phi_a$ is contiguous along $p_2$ dimension if we have $(p_1, y_1) \in \Phi_a$ and $(p_1,y_2) \in \Phi_a$, then $\forall y \in [y_1,y_2]$, we have $(p_1,y) \in \Phi_a$.
\end{definition}
\begin{thm}
\label{thm:contiguous}
$\Phi_{B_b}$ is contiguous in both $p_1$ and $p_2$ dimensions. $\Phi_{B_1}$ is contiguous in $p_1$ dimension, and $\Phi_{B_2}$ is contiguous in $p_2$ dimension.
\end{thm}
\begin{proof}
Here we will prove the theory for $\Phi_{B1}$, and the results for $\Phi_{B_2}$ and $\Phi_{B_b}$ can be proved in a similar manner.
Let $(x_1,p_2), (x_2,p_2) \in \Phi_{B_1}$, next we show that $((c x_1+(1-c)x_2),p_2))$ is also in region $\Phi_{B_1}$, where $c \in [0,1]$.
\begin{eqnarray}
\label{eq:vbx12}
& & V ((c x_1+(1-c)x_2),p_2) \nonumber \\
& \leq & c V(x_1,p_2) + (1-c) V(x_2,p_2)\nonumber \\
&=&c V_{B_1}(x_1,p_2) + (1-c) V_{B_1}(x_2,p_2) \nonumber\\
&=&V_{B_1} ((c x_1+(1-c)x_2),p_2)\nonumber \\
& \leq & V ((c x_1+(1-c)x_2),p_2),
\end{eqnarray}
where the first inequality comes from the convexity of $V(p_1,p_2)$ in $p_1$; the first equality follows from the fact that $(x_1,p_2), (x_2,p_2) \in \Phi_{B_1}$; the second equality comes from the fact that $V_{B_1}$ is linear in $p_1$ as in Lemma \ref{lem:linear}; the last inequality follows from the definition of $V(p_1,p_2)$. In the above equation, we have $V ((c x_1+(1-c)x_2),p_2)=V_{B_1} ((c x_1+(1-c)x_2),p_2)$, which means $(c x_1+(1-c)x_2),p_2)$ is in the region $\Phi_{B1}$, therefore $\Phi_{B1}$ is contiguous in $p_1$ dimension by definition \ref{def:contiguous}.
\end{proof}

\begin{thm}
\label{thm:phib1b2sym}
If belief $(p_1,p_2)$ is in $\Phi_{B_1}$, then belief $(p_2,p_1)$ is in $\Phi_{B_2}$. In other words, the decision regions of $B_1$ and $B_2$ are mirrors with respect to the line $p_1=p_2$ in the belief space.
\end{thm}
\begin{proof}
Let $(p_1,p_2)$ be a belief state in the decision region of $B_1$, then we have
\begin{eqnarray}
V(p_1,p_2)&=& \max \{V_{B_b}(p_1,p_2),V_{B_1}(p_1,p_2),V_{B_2}(p_1,p_2)\} \nonumber \\
& =&V_{B_1}(p_1,p_2).
\end{eqnarray}
Using equations (\ref{eq:vbbb}),(\ref{eq:vbb1}) and (\ref{eq:vbb2}),we have
\begin{eqnarray}
\label{eq:vbb1xy}
& & V_{B_1}(p_1,p_2) \nonumber \\
&=&p_1R_h+\beta[(1-p_1)V(\lambda_0,T(p_2))+p_1V(\lambda_1,T(p_2))] \nonumber\\
& > & V_{B_2}(p_1,p_2) \nonumber \\
& = & p_2R_h+\beta[(1-p_2)V(T(p_1),\lambda_0)+p_2V(T(p_1),\lambda_1)], \nonumber \\
\end{eqnarray}
and
\begin{eqnarray}
\label{eq:vbb1xyvb}
& &V_{B_1}(p_1,p_2) \nonumber \\
&> & V_{B_b}(p_1,p_2) \nonumber \\
& = &p_1R_l+p_2R_l+\beta[(1-p_1)(1-p_2)V(\lambda_0,\lambda_0)\nonumber \\
&+&p_1(1-p_2)V(\lambda_1,\lambda_0)+(1-p_1)p_2V(\lambda_0,\lambda_1) \nonumber\\
&+&p_1p_2V(\lambda_1,\lambda_1)].
\end{eqnarray}
Now consider the belief state of $(p_2,p_1)$,
\begin{eqnarray}
& & V_{B_2}(p_2,p_1) \nonumber \\
&=&p_1R_h+\beta[(1-p_1)V(T(p_2),\lambda_0)+p_1V(T(p_2),\lambda_1)] \nonumber \\
& =& V_{B_1}(p_1,p_2) \nonumber \\
& > & V_{B_2}(p_1,p_2) \nonumber \\
& =&p_2R_h+\beta[(1-p_2)V(T(p_1),\lambda_0)+p_2V(T(p_1),\lambda_1)] \nonumber \\
& = & V_{B_1}(p_2,p_1),
\end{eqnarray}
where the second and last equations follow by comparing the expression in equation (\ref{eq:vbb1xy}) and using the fact that $V(p_1,p_2)=V(p_2,p_1)$ (Lemma \ref{lem:vbsem}). Similarly, from (\ref{eq:vbb1xyvb}) and Lemma \ref{lem:vbsem}, we have
\begin{equation}
V_{B_2}(p_2,p_1)>V_{B_b}(p_2,p_1).
\end{equation}
Thus we have
\begin{eqnarray}
V(p_2,p_1)&=&\max \{V_{B_b}(p_2,p_1),V_{B_1}(p_2,p_1),V_{B_2}(p_2,p_1)\} \nonumber \\
&=&V_{B_2}(p_2,p_1),
\end{eqnarray}
which means $(p_2,p_1)$ lies in the decision region of $B_2$, that is, $(p_2,p_1) \in \Phi_{B_2}$. This concludes the proof.
\end{proof}
\begin{thm}
\label{thm:PhiBbsym}
If belief $(p_1,p_2)$ is in $\Phi_{B_b}$, then belief $(p_2,p_1)$ is in $\Phi_{B_b}$. That is, the decision region of $B_b$ is symmetric with respect to the line $p_1=p_2$ in the belief region.
\end{thm}
\begin{proof}
Suppose $(p_1,p_2)$ is in $\Phi_{B_b}$, then we have
\begin{eqnarray}
V(p_1,p_2)&=&\max\{V_{B_b}(p_1,p_2),V_{B_1}(p_1,p_2),V_{B_2}(p_1,p_2)\} \nonumber \\
&=&V_{B_b}(p_1,p_2)
\end{eqnarray}
Now consider the belief state $(p_2,p_1)$,
\begin{eqnarray}
& &V_{B_b}(p_2,p_1) \nonumber \\
&= &p_2R_l+p_1R_l+\beta[(1-p_2)(1-p_1)V(\lambda_0,\lambda_0)\nonumber \\
&+&p_2(1-p_1)V(\lambda_1,\lambda_0)+(1-p_2)p_1V(\lambda_0,\lambda_1) \nonumber\\
&+&p_2p_1V(\lambda_1,\lambda_1)] \nonumber \\
&=& V_{B_b}(p_1,p_2) \nonumber \\
&>& V_{B_1}(p_1,p_2)\nonumber \\
&=&p_1R_h+\beta[(1-p_1)V(\lambda_0,T(p_2))+p_1V(\lambda_1,T(p_2))] \nonumber \\
&=& V_{B_2}(p_2,p_1),
\end{eqnarray}
where the equations follow from (\ref{eq:vbbb}), (\ref{eq:vbb1}), (\ref{eq:vbb2}) and Lemma \ref{lem:vbsem}. The inequality comes from the assumption that $(p_1,p_2)$ is in $\Phi_{B_b}$. Similarly, we have $V_{B_b}(p_2,p_1) > V_{B_1}(p_2,p_1)$. That is,
\begin{eqnarray}
V(p_2,p_1)&=&\max\{V_{B_b}(p_2,p_1),V_{B_1}(p_2,p_1),V_{B_2}(p_2,p_1)\} \nonumber \\
&=&V_{B_b}(p_2,p_1).
\end{eqnarray}
That is, $(p_2,p_1)$ is in $\Phi_{B_b}$. And this concludes the proof.
\end{proof}
\begin{lem}
After each channel is used once, the belief state is the four sides of a rectangle determined by four vertices at $(\lambda_0,\lambda_0), (\lambda_0,\lambda_1), (\lambda_1,\lambda_0), (\lambda_1,\lambda_1)$ (Figure \ref{fig:decisionline} (a)).
\end{lem}
\begin{proof}
From the belief update in (\ref{eq:vbbb})(\ref{eq:vbb1})(\ref{eq:vbb2}), it is clear that the belief state of a channel is updated to one of the following three values after any action: $\lambda_0$, $\lambda_1$, or $T(p)$, where $p$ is the current belief of a channel. For any $0 \leq p \leq 1$, $\lambda_0 \leq T(p)= \lambda_0 + (\lambda_1-\lambda_0)p \leq \lambda_1$. Therefore the belief state of a channel is between $\lambda_0$ and $\lambda_1$.

Furthermore, since at least one channel is used in our power allocation strategy, its channel state is revealed at the end of the time slot. This means at least one of the channel has a belief of either $\lambda_0$ or $\lambda_1$. And this concludes the proof.
\end{proof}

\begin{figure}
  \centering
  \includegraphics[width=3.5in]{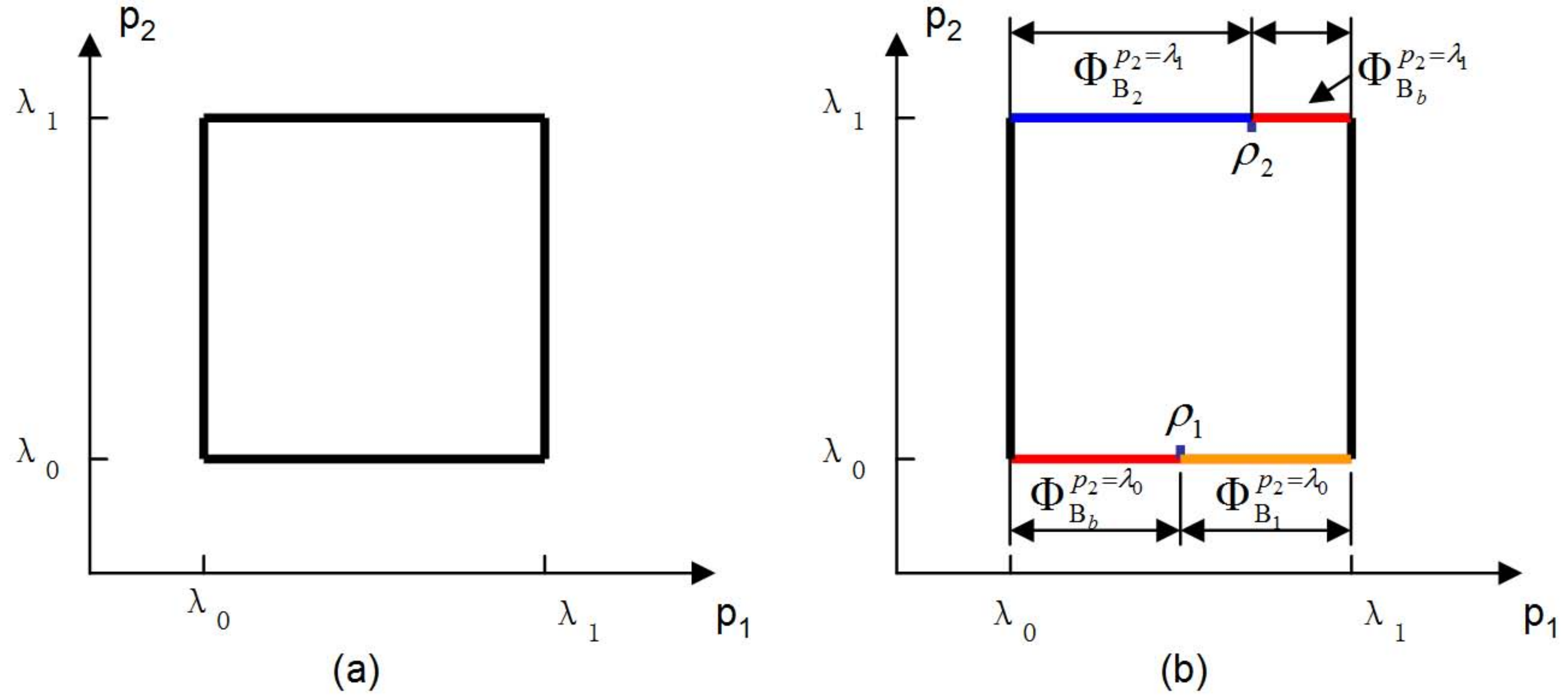}
  \caption{(a) The feasible belief space. (b) The threshold on $p_1$ ( $p_2=\lambda_0(\lambda_1$)).}
  \label{fig:decisionline}
\end{figure}
\begin{thm}
\label{thm:rho1}
Let $p_1 \in [\lambda_0, \lambda_1]$, $p_2=\lambda_0$, there exists a threshold $\rho_1$($\lambda_0 \leq \rho_1 \leq \lambda_1$) such that
$\forall p_1 \in [\lambda_0, \rho_1], (p_1, \lambda_0) \in \Phi_{B_b}$. (Figure \ref{fig:decisionline}(b))
\end{thm}

\begin{proof}
We introduce the following sets
\begin{eqnarray}
\Phi_a^{p_2=\lambda_0} &=&\{(p_1 \in [\lambda_0, \lambda_1], \lambda_0), V(p_1,\lambda_0)=V_a(p_1,\lambda_0)\}, \nonumber \\
& & a \in \{B_b, B_1, B_2\}.
\end{eqnarray}
We will first prove that $\Phi_{B_b}^{p_2=\lambda_0}$ and $\Phi_{B_1}^{p_2=\lambda_0}$ are convex, which is important to prove the structure of the optimal policy. When $p_2=\lambda_0$, $V_{B_b}(p_1,p_2)$ is rewritten as
\begin{eqnarray}
\label{eq:vbbbyeqlam0}
& &V_{B_b}(p_1,\lambda_0) \nonumber \\
&=&p_1R_l+\lambda_0R_l+\beta[(1-p_1)(1-\lambda_0)V(\lambda_0,\lambda_0)\nonumber \\
&+&p_1(1-\lambda_0)V(\lambda_1,\lambda_0)+(1-p_1)\lambda_0V(\lambda_0,\lambda_1) \nonumber\\
&+&p_1\lambda_0V(\lambda_1,\lambda_1)] \nonumber\\
&=& p_1[R_l-\beta(1-\lambda_0)V(\lambda_0,\lambda_0)-\beta\lambda_0V(\lambda_1,\lambda_0)\nonumber\\
&-&\beta \lambda_0V(\lambda_0,\lambda_1)+\lambda_0V(\lambda_1,\lambda_1)]+\lambda_0R_l\nonumber\\
&+&\beta(1-\lambda_0)V(\lambda_0,\lambda_0)+\beta\lambda_0V(\lambda_0,\lambda_1).
\end{eqnarray}
From equation (\ref{eq:vbbbyeqlam0}) it is easy to see that $V_{B_b}(p_1,\lambda_0)$ is linear in $p_1$. Let $(x_1,\lambda_0),(x_2,\lambda_0) \in \Phi_{B_b}^{p_2=\lambda_0}$ and let $c\in[0,1]$ then we have
\begin{eqnarray}
& &V(cx_1+(1-c)x_2, \lambda_0) \nonumber \\
&\leq& cV(x_1,\lambda_0)+(1-c)V(x_2,\lambda_0)\nonumber \\
&=&cV_{B_b}(x_1,\lambda_0)+(1-c)V_{B_b}(x_2,\lambda_0)\nonumber \\
&=&V_{B_b}(ax_1+(1-c)x_2,\lambda_0) \nonumber \\
&\leq& V(cx_1+(1-c)x_2,\lambda_0),
\end{eqnarray}
where the first inequality comes from the convexity of $V(p_1,\lambda_0)$; the first equality follows from the fact that $(x_1,\lambda_0),(x_2,\lambda_0) \in \Phi_{B_b}^{p_2=\lambda_0}$, and the second equality from the linearity of $V_{B_b}(p_1,\lambda_0)$; the last inequality comes from the definition of $V(\cdot)$. Consequently, $V_{B_b}(cx_1+(1-c)x_2,\lambda_0)=V(cx_1+(1-c)x_2,\lambda_0)$, hence $(cx_1+(1-c)x_2,\lambda_0) \in \Phi_{B_b}^{p_2=\lambda_0}$, which proves the convexity of $\Phi_{B_b}^{p_2=\lambda_0}$. Since convex subsets of the real line are intervals and $(\lambda_0, \lambda_0) \in \Phi_{B_b}^{p_2=\lambda_0}$ (from the fact that $\Phi_{B_b}$ is symmetric), there exists $\rho_1 \in [\lambda_0, \lambda_1]$ such that $\Phi_{B_b}^{p_2=\lambda_0}=[\lambda_0..\rho_1,\lambda_0$] (Figure \ref{fig:decisionline}(b)). In other words, $\forall p_1 \in [\lambda_0, \rho_1], (p_1, \lambda_0) \in \Phi_{B_b}$. And this concludes the proof.

\end{proof}

\begin{thm}
\label{thm:rho2}
Let $p_1 \in [\lambda_0, \lambda_1]$, $p_2=\lambda_1$, there exists a threshold $\rho_2$($\lambda_0 \leq \rho_2 \leq \lambda_1$) such that
$\forall p_1 \in [\rho_2,\lambda_1], (p_1, \lambda_1) \in \Phi_{B_b}$. (Figure \ref{fig:decisionline}(b))
\end{thm}
\begin{proof}
Similar to proof of theorem \ref{thm:rho1}, we can show that $\Phi_{B_b}^{p_2=\lambda_1}$ is convex, therefore it is an interval on $p_2=\lambda_1$. Now since $(\lambda_1, \lambda_1)$ is in $\Phi_{B_b}$ from the fact that $\Phi_{B_b}$ is symmetric, there exists a threshold $\rho_2$, such that $\forall p_1 \in [\rho_2,\lambda_1], (p_1, \lambda_1) \in \Phi_{B_b}$.
\end{proof}

\begin{lem}
\label{lem:p2lambda0}
In case of $p_2=\lambda_0$, it is not optimal to take action $B_2$. In case of $p_2=\lambda_1$, it is not optimal to take action $B_1$.
\end{lem}
\begin{proof}
In case of $p_2=\lambda_0$, we need to prove that it is not optimal to take action $B_2$, \textit{i.e.}
\begin{align}\label{eq:ineq}
V_{B_2}(p_1,\lambda_0) &\leq V_{B_b}(p_1,\lambda_0) \ \  \text{or} \nonumber\\
V_{B_2}(p_1,\lambda_0)&\leq V_{B_1}(p_1,\lambda_0).
\end{align}
If one of the above inequalities holds, then the proof is complete. Because among three options, $B_2$ would be second or third then it's not optimal.
We will prove the first inequality as follows:
\begin{align}\label{eq:vb1}
V_{B_2}(p_1,\lambda_0)&=R_h \lambda_0+\beta \lambda_0 V(T(p_1),\lambda_1) \nonumber\\
&+\beta (1-\lambda_0)V(T(p_1),\lambda_0). \nonumber\\
V_{B_b}(p_1,\lambda_0)&=R_l \lambda_0+R_l p_1+\beta \lambda_0 p_1 V(\lambda_1,\lambda_1) \nonumber\\
&+\beta (1-\lambda_0)p_1 V(\lambda_1,\lambda_0) +\beta \lambda_0(1-p_1) V(\lambda_0,\lambda_1) \nonumber\\
&+\beta (1-\lambda_0)(1-p_1) V(\lambda_0,\lambda_0).
\end{align}
Then we have:
\begin{eqnarray}\label{eq:dif}
& &V_{B_b}(p_1,\lambda_0)-V_{B_2}(p_1,\lambda_0)\nonumber\\
&=&[R_l p_1-(R_h-R_l)\lambda_0] \nonumber\\
&+& \beta \lambda_0[p_1 V(\lambda_1,\lambda_1)+(1-p_1) V(\lambda_0,\lambda_1)-V(T(p_1),\lambda_1)] \nonumber\\
&+&\beta (1-\lambda_0)[p_1 V(\lambda_1,\lambda_0)+(1-p_1) V(\lambda_0,\lambda_0)\nonumber \\
&-&V(T(p_1),\lambda_0)]. \nonumber\\
\end{eqnarray}
For the first term of (\ref{eq:dif}) we have:
\begin{align}\label{eq:rl}
&R_l p_1-(R_h-R_l) \lambda_0\geq \lambda_0[2R_l-R_h]\geq 0.
\end{align}
In the above inequality, we use the fact that $p_1\geq \lambda_0$ and $R_h < 2R_l$.

Assume that at point $(T(p_1),\lambda_1)$, the action $B_i, i\in \{1,2,b\}$ is optimal. Then for the second term of (\ref{eq:dif}), we have:
\begin{eqnarray}\label{eq:l0}
& &p_1 V(\lambda_1,\lambda_1)+(1-p_1) V(\lambda_0,\lambda_1)-V(T(p_1),\lambda_1) \nonumber\\
&=&p_1 V(\lambda_1,\lambda_1)+(1-p_1) V(\lambda_0,\lambda_1)-V_{B_i}(T(p_1),\lambda_1) \nonumber\\
&\geq &p_1 V_{B_i}(\lambda_1,\lambda_1)+(1-p_1) V_{B_i}(\lambda_0,\lambda_1)-V_{B_i}(T(p_1),\lambda_1) \nonumber\\
&=& V_{B_i}(p_1 \lambda_1+(1-p_1)\lambda_0,\lambda_1)-V_{B_i}(T(p_1),\lambda_1) \nonumber\\
&=&V_{B_i}(T(p_1),\lambda_1)-V_{B_i}(T(p_1),\lambda_1)=0. \nonumber\\
\end{eqnarray}
The first inequality above is achieved from the fact that $V\geq V_{B_i}, i\in \{1,2,b\}$ and the equality after the inequality is from the linearity of $V_{B_i}, i\in \{1,2,b\}$ as in Lemma 1.

Similarly, for the third term of (\ref{eq:dif}), assume that at point $(T(p_1),\lambda_0)$, the action $B_j$ is optimal. Then we have:
\begin{eqnarray}\label{eq:1-l0}
& &p_1 V(\lambda_1,\lambda_0)+(1-p_1) V(\lambda_0,\lambda_0)-V(T(p_1),\lambda_0)\nonumber\\
&=&p_1 V(\lambda_1,\lambda_0)+(1-p_1) V(\lambda_0,\lambda_0)-V_{B_j}(T(p_1),\lambda_0)\nonumber\\
&\geq& p_1 V_{B_j}(\lambda_1,\lambda_0)+(1-p_1) V_{B_j}(\lambda_0,\lambda_0)-V_{B_j}(T(p_1),\lambda_0) \nonumber\\
&=& V_{B_j}(p_1 \lambda_1+(1-p_1)\lambda_0,\lambda_0)-V_{B_j}(T(p_1),\lambda_0) \nonumber\\
&=&V_{B_j}(T(p_1),\lambda_0)-V_{B_j}(T(p_1),\lambda_0)=0. \nonumber\\
\end{eqnarray}
Now using (\ref{eq:rl}), (\ref{eq:l0}) and (\ref{eq:1-l0}) in (\ref{eq:dif}), we have:
\begin{equation}\label{eq:pos}
V_{B_b}(p_1,\lambda_0)-V_{B_2}(p_1,\lambda_0)\geq 0.
\end{equation}
(\ref{eq:pos}) means that $B_2$ never can be optimal on the border of  $(p_1,\lambda_0)$.

Similar arguments can be used to prove that $V_{B_1}(p_1,\lambda_1) \leq V_{B_b}(p_1, \lambda_1)$, thus $B_1$ is not optimal on the border of $(p_1, \lambda_1)$. Then the proof is complete.
\end{proof}
\subsection*{C: The structure of the optimal policy}

\begin{thm}
\label{thm:structure}
The optimal policy has a simple threshold structure and can be described as follows (Figure \ref{fig:optimalpolicy}):
\[ \pi^*(p_1,\lambda_0)= \left \{
\begin{array}{rl}
B_b, & \quad \mbox{if} \quad \lambda_0 \leq p_1 \leq \rho_1\\
B_1, & \quad \mbox{if} \quad   \rho_1 < p_1 \leq \lambda_1\\
\end{array} ,   \ \ \ \   (a) \right. \]
\[ \pi^*(p_1,\lambda_1)= \left \{
\begin{array}{rl}
B_b, & \quad \mbox{if} \quad \rho_2 \leq p_1 \leq \lambda_1\\
B_2, & \quad \mbox{if} \quad   \lambda_0 \leq p_1 < \rho_2\\
\end{array} , \ \ \ \   (b)  \right. \]
\[ \pi^*(\lambda_0,p_2)= \left \{
\begin{array}{rl}
B_b, & \quad \mbox{if} \quad \lambda_0 \leq p_2 \leq \rho_1\\
B_2, & \quad \mbox{if} \quad   \rho_1 < p_2 \leq \lambda_1\\
\end{array} , \ \ \ \    (c)  \right. \]
\[ \pi^*(\lambda_1,p_2)= \left \{
\begin{array}{rl}
B_b, & \quad \mbox{if} \quad \rho_2 \leq p_2 \leq \lambda_1\\
B_1, & \quad \mbox{if} \quad   \lambda_0 \geq p_2 < \rho_1\\
\end{array} . \ \ \ \   (d) \right. \]
\begin{equation}\label{eqn:structure}
\end{equation}
\end{thm}
\begin{proof}Let us first consider the border of $(p_1, \lambda_0)$. From Lemma \ref{lem:p2lambda0} we understand that on this border $B_2$ is not optimal, therefore the optimal action can only be $B_b$ or $B_1$. Furthermore from Theorem \ref{thm:rho1} we know that the decision region on this border for $B_b$ is the interval represented by $\lambda_0 \leq  p_1 \leq \rho_1$, it follows directly that the remaining part of this border must belong to the decision region of $B_1$. Thus we have (\ref{eqn:structure})(a).

(\ref{eqn:structure})(b) specifies the optimal action on the border of $(p_1, \lambda_1)$. Similar to the (\ref{eqn:structure})(a), it is directly obtained from Theorem \ref{thm:rho2} and Lemma \ref{lem:p2lambda0}.

(\ref{eqn:structure})(c) specifies the optimal action on the border of $(\lambda_0, p_2)$. It is directly obtained based on the result on the border of $(p_1, \lambda_0)$ (i.e. (\ref{eqn:structure}) (a)) using Theorems \ref{thm:phib1b2sym} and \ref{thm:PhiBbsym}. Specifically, from Theorem \ref{thm:PhiBbsym} we know that the decision region of $B_b$ is symmetric with respect to the line $p_1=p_2$, thus we have the first term of (\ref{eqn:structure})(c) from the first term of (\ref{eqn:structure})(a). Similarly, from Theorem \ref{thm:phib1b2sym} we know the decision regions of $B_1$ and $B_2$ are mirrors with respect to the line $p_1=p_2$, therefore we get the second term of (\ref{eqn:structure})(c) from the second term of (\ref{eqn:structure})(a).

(\ref{eqn:structure})(d) specifies the optimal action on the border of $(\lambda_1, p_2)$. It is directly obtained based on the result on the border of $(p_1, \lambda_1)$ (i.e. (\ref{eqn:structure}) (b)) using Theorems \ref{thm:phib1b2sym} and \ref{thm:PhiBbsym}.
\end{proof}

\begin{figure}[]
  \centering
  \includegraphics[width=3.5in]{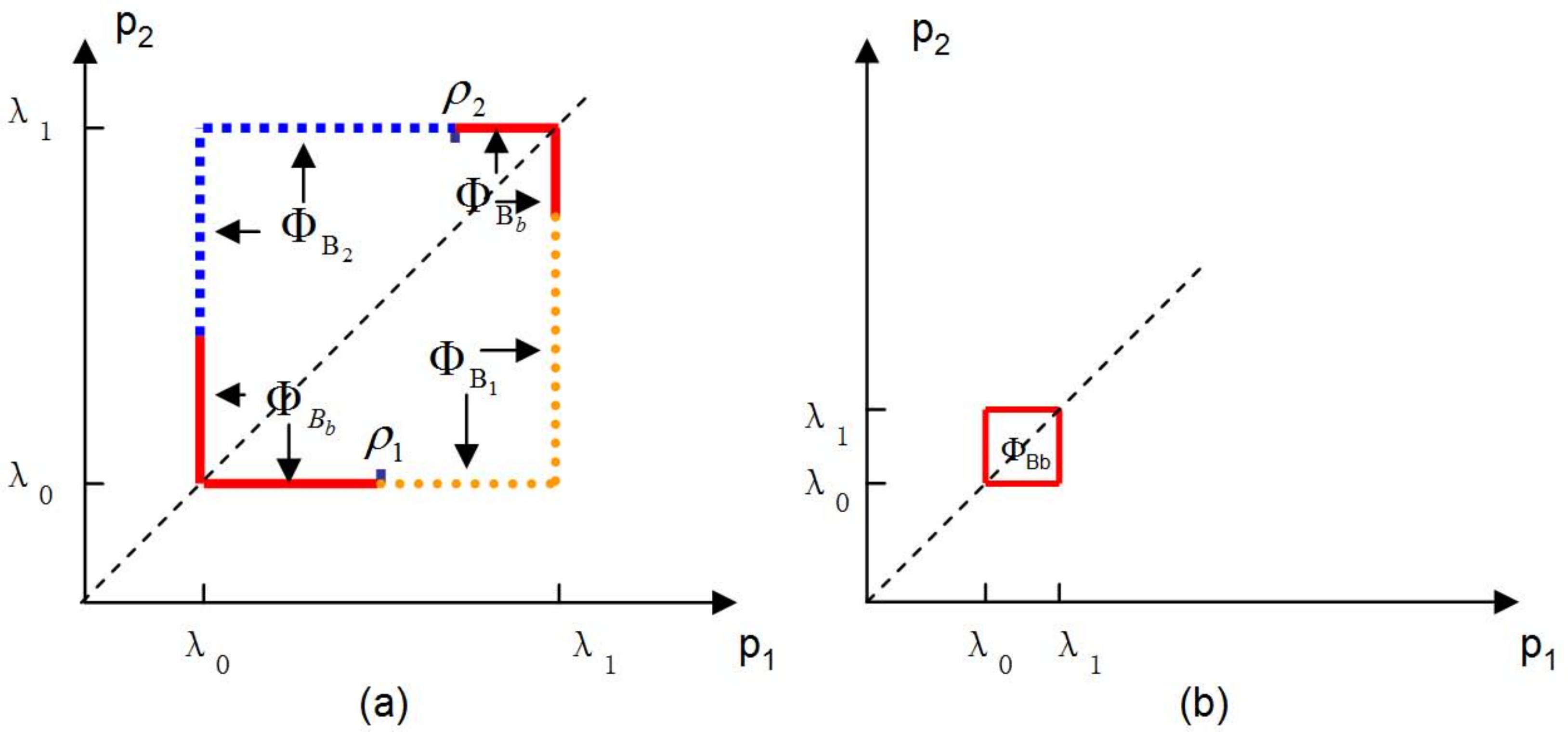}
  \caption{Structure of optimal policy.}
  \label{fig:optimalpolicy}
\end{figure}


From the above analysis we understand that the optimal policy has a simple threshold structure. And it is critical to find the two thresholds $\rho_1$ and $\rho_2$.

\begin{thm}\label{thm:rho1calculation}
Let $\delta_{i,j}(k_1,k_2)=V_{B_i}(k_1,k_2)-V_{B_j}(k_1,k_2), (i \in \{1,2,b\}, j \in \{1,2,b \})$, $\rho_1$ can be calculated as follows

1) if $T(\lambda_0) < \rho_2$, $T(\lambda_0) \leq \rho_1$
\begin{equation}\label{eq:rho1cal1}
\rho_1=\frac{\lambda_0 R_l+\beta \lambda_0 \delta_{2,b}(\lambda_0, \lambda_1)}{R_h-R_l+\beta \lambda_0 (\delta_{1,b}(\lambda_1,\lambda_1)+\delta_{2,b}(\lambda_0, \lambda_1))},
\end{equation}

2) if $T(\lambda_0) < \rho_2$, $T(\lambda_0) > \rho_1$
\begin{equation}\label{eq:rho1cal2}
\rho_1=\frac{\lambda_0 R_l+\beta (1-\lambda_0) \delta_{b,2}(\lambda_0, \lambda_0)}{R_h-R_l+\beta \lambda_0 \delta_{1,b}(\lambda_1,\lambda_1)+ \beta (1-\lambda_0)\delta_{b,2}(\lambda_0, \lambda_0)},
\end{equation}

3) if $T(\lambda_0) \geq \rho_2$, $T(\lambda_0) \leq \rho_1$
\begin{equation}\label{eq:rho1cal3}
\rho_1=\frac{\lambda_0 R_l+\beta \lambda_0 \delta_{2,b}(\lambda_0, \lambda_1)}{R_h-R_l+\beta \lambda_0 \delta_{2,b}(\lambda_0,\lambda_1)+ \beta (1-\lambda_0)\delta_{b,1}(\lambda_1, \lambda_0)},
\end{equation}

4) if $T(\lambda_0) \geq \rho_2$, $T(\lambda_0) > \rho_1$, $\rho_1$ is calculated in (\ref{eq:rho1cal4}).
\begin{figure*}
\begin{equation}\label{eq:rho1cal4}
\rho_1=\frac{\lambda_0 R_l+\beta \lambda_0 \delta_{2,1}(\lambda_0, \lambda_1)+\beta (1-\lambda_0) \delta_{b,1}(\lambda_0, \lambda_0)}{R_h-R_l+\beta \lambda_0 \delta_{2,1}(\lambda_0,\lambda_1)+ \beta (1-\lambda_0)(\delta_{b,1}(\lambda_1, \lambda_0)+\delta_{b,1}(\lambda_0, \lambda_0))}.
\end{equation}
\end{figure*}
\end{thm}
\begin{proof}
We will prove (\ref{eq:rho1cal1}) and the rest of the theorem can be shown in a similar manner. From Theorem \ref{thm:structure} we know that at the point $(\rho_1, \lambda_0)$
\begin{equation}
V_{B_1}(\rho_1, \lambda_0) = V_{B_b}(\rho_1, \lambda_0).
\end{equation}
Using (\ref{eq:vbb1}) and (\ref{eq:vbbb}), the above is written as
\begin{eqnarray}\label{eq:rho1proof1}
&&\rho_1 R_h + \rho_1 \beta V(\lambda_1, T(\lambda_0))+(1-\rho_1)\beta V(\lambda_0, T(\lambda_0))= \nonumber \\
&&\rho_1 R_l + \lambda_0 R_l + \beta \rho_1 \lambda_0 V(\lambda_1, \lambda_1)+\beta (1 -\rho_1) \lambda_0 V(\lambda_0, \lambda_1) \nonumber \\
&&+\beta \rho_1 (1-\lambda_0) V(\lambda_1, \lambda_0)+\beta (1- \rho_1)(1- \lambda_0) V(\lambda_0, \lambda_0). \nonumber \\
\end{eqnarray}
From Theorem \ref{thm:structure} and the condition that $T(\lambda_0) < \rho_2$, $T(\lambda_0) \leq \rho_1$, we have $V(\lambda_1,T(\lambda_0))=V_{B_1}(\lambda_1,T(\lambda_0))$, $V(\lambda_0,T(\lambda_0))=V_{B_b}(\lambda_0,T(\lambda_0))$ under the assumption that $\rho_1 \geq T(\lambda_0)$, $V(\lambda_1,\lambda_1)=V_{B_b}(\lambda_1,\lambda_1)$, $V(\lambda_0,\lambda_1)=V_{B_2}(\lambda_0,\lambda_1)$, $V(\lambda_1,\lambda_0)=V_{B_1}(\lambda_1,\lambda_0)$,  $V(\lambda_0,\lambda_0)=V_{B_b}(\lambda_0,\lambda_0)$.

Thus (\ref{eq:rho1proof1}) can be written as
\begin{align}\label{eq:rho1proof2}
\rho_1 R_h + \rho_1 \beta V_{B_1}(\lambda_1, T(\lambda_0))+(1-\rho_1)\beta V_{B_b}(\lambda_0, T(\lambda_0))= \nonumber \\
\rho_1 R_l + \lambda_0 R_l + \beta \rho_1 \lambda_0 V_{B_b}(\lambda_1, \lambda_1)+\beta (1 -\rho_1) \lambda_0 V_{B_2}(\lambda_0, \lambda_1) \nonumber \\
+\beta \rho_1 (1-\lambda_0) V_{B_1}(\lambda_1, \lambda_0)+\beta (1- \rho_1)(1- \lambda_0) V_{B_b}(\lambda_0, \lambda_0). \nonumber \\
\end{align}

Using the linearity of $V_{B_1}$ and $V_{B_b}$ in $p_2$, and the fact that $T(\lambda_0)=\lambda_0 \lambda_1 + (1-\lambda_0) \lambda_0$, we have
\begin{eqnarray}\label{eq:rho1proof3}
V_{B_1}(\lambda_1, T(\lambda_0))=\lambda_0 V_{B_1}(\lambda_1,\lambda_1)+(1-\lambda_0)V_{B_1}(\lambda_1, \lambda_0),\nonumber \\
V_{B_b}(\lambda_0, T(\lambda_0))=\lambda_0 V_{B_b}(\lambda_0,\lambda_1)+(1-\lambda_0)V_{B_b}(\lambda_0, \lambda_0).
\end{eqnarray}

Substitute $V_{B_1}(\lambda_1, T(\lambda_0))$ and $V_{B_b}(\lambda_0, T(\lambda_0))$ in (\ref{eq:rho1proof2}) with (\ref{eq:rho1proof3}), with simple manipulation we have formula (\ref{eq:rho1cal1}).
\end{proof}

\begin{thm}\label{thm:rho2calculation}
Let $\delta_{i,j}(k_1,k_2)=V_{B_i}(k_1,k_2)-V_{B_j}(k_1,k_2), (i \in \{1,2,b\}, j \in \{1,2,b \})$, the threshold $\rho_2$ is calculated as follows

1) if $T(\rho_2) \geq \rho_2 $ and $T(\rho_2) > \rho_1$
\begin{equation}
\rho_2=\frac{\lambda_1(R_h-R_l)-\beta \lambda_1 \delta_{2,b}(\lambda_0,\lambda_1)-\beta(1-\lambda_1)\delta_{b,1}(\lambda_0,
\lambda_0)}{R_l - \beta \lambda_1 \delta_{2,b}(\lambda_0,\lambda_1)-\beta(1-\lambda_1)\delta_{b,1}(\lambda_0, \lambda_0)},
\end{equation}

2)  if $T(\rho_2) \geq \rho_2 $ and $T(\rho_2) \leq \rho_1$
\begin{equation}
\rho_2=\frac{\lambda_1(R_h-R_l)-\beta \lambda_1 \delta_{2,b}(\lambda_0,\lambda_1))}{R_l - \beta \lambda_1 \delta_{2,b}(\lambda_0,\lambda_1)-\beta(1-\lambda_1)\delta_{b,1}(\lambda_1, \lambda_0)},
\end{equation}

3) if $T(\rho_2) < \rho_2$, $T(\rho_2) > \rho_1$
\begin{equation}
\rho_2=\frac{\lambda_1(R_h-R_l)-\beta (1-\lambda_1) \delta_{b,1}(\lambda_0,\lambda_0))}{R_l - \beta \lambda_1 \delta_{2,b}(\lambda_1,\lambda_1)-\beta(1-\lambda_1)\delta_{b,1}(\lambda_0, \lambda_0)},
\end{equation}

4) if $T(\rho_2) < \rho_2$, $T(\rho_2) \leq \rho_1$
\begin{equation}
\rho_2=\frac{\lambda_1(R_h-R_l)}{R_l - \beta \lambda_1 \delta_{2,b}(\lambda_1,\lambda_1)-\beta(1-\lambda_1)\delta_{b,1}(\lambda_1, \lambda_0)}.
\end{equation}

\end{thm}
The proof of this theorem is similar to that of theorem \ref{thm:rho1calculation} and is omitted here.


\section{Simulation based on linear programming}

Linear programming is one of the approaches to solve the Bellman's equation in (\ref{eq:bellman}). Based on \cite{farias2002}, we model our problem as the following linear program:

\begin{center}
\begin{eqnarray}
\min \sum_{\vec{p} \in \mathbb{X}}{V(\vec{p})}, & &\nonumber \\
\text{s.t. \ \ }  g_a(\vec{p})&+& \beta \sum_{\vec{y} \in \mathbb{X}}f_a(\vec{p},\vec{y})V(\vec{y})  \leq  V(\vec{p}), \nonumber\\
& & \forall \vec{p} \in \mathbb{X},  \forall a \in \mathbb{A}_{\vec{p}},
\label{eqn:vlp}
\end{eqnarray}
\end{center}
where $\mathbb{X}$ is the space of belief state, $\mathbb{A}_{\vec{p}}$ is the set of available actions for state $\vec{p}$. The state transition probabilities $f_a(\vec{p},\vec{y})$ is the probability that the next state will be $\vec{y}$ given that the current state is $\vec{p}$ and the current action is $a \in \mathbb{A}_{\vec{p}}$. The optimal policy can be generated according to
\begin{equation}
\pi(\vec{p})=\arg \max_{a \in \mathbb{A}_{\vec{p}}}(g_a(\vec{p})+\beta\sum_{\vec{y} \in \mathbb{X}}f_a(\vec{p},\vec{y})V(\vec{y})).
\label{eqn:policy}
\end{equation}

We used the LOQO solver on NEOS Server \cite{neos} with AMPL input \cite{ampl} to obtain the solution of equation (\ref{eqn:vlp}). Then we used MATLAB to construct the policy according to equation (\ref{eqn:policy}).

\begin{figure}
  \centering
  \includegraphics[width=2.5in]{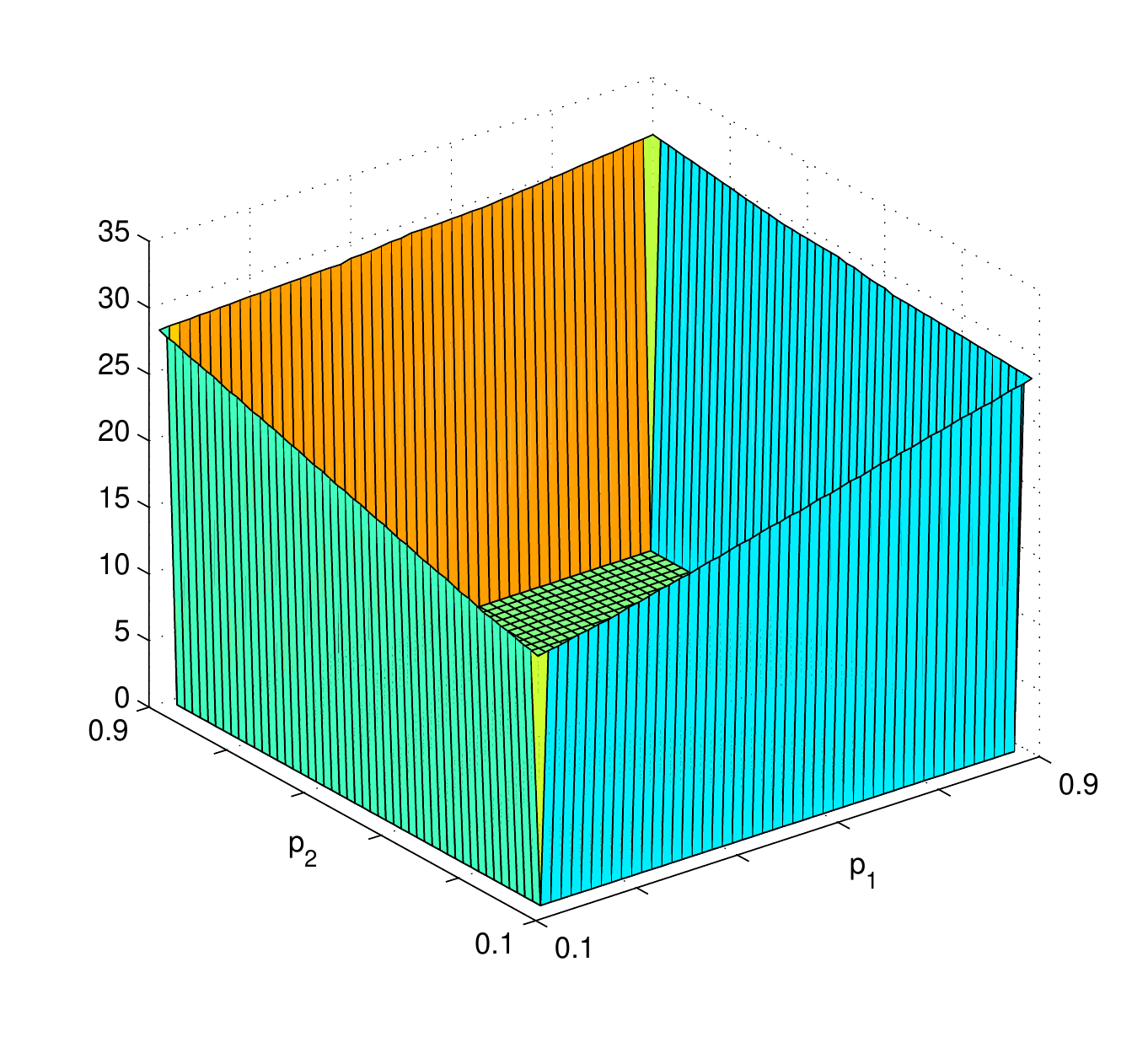}
  \caption{Value function.}
  \label{fig:value}
\end{figure}

\begin{figure}
  \centering
  \includegraphics[width=2in]{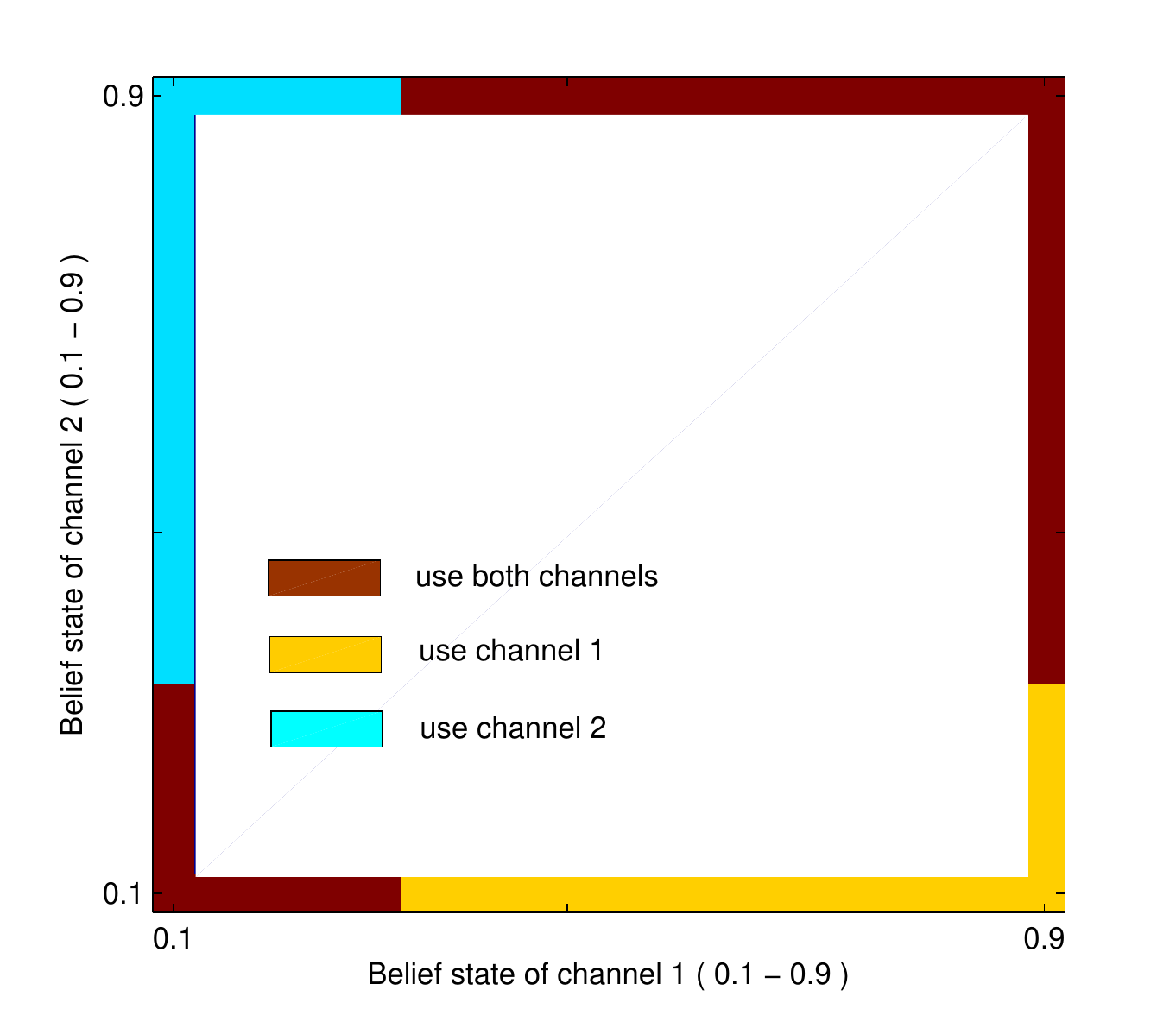}
  \caption{Optimal policy.}
  \label{fig:policy}
\end{figure}

Figure \ref{fig:value} shows the AMPL solution of value function for the following set of parameters: $\lambda_0=0.1, \lambda_1=0.9, \beta = 0.9, R_l=2, R_h=3$. The corresponding optimal policy is shown in Figure \ref{fig:policy}. The structure of the policy in Figure \ref{fig:policy} clearly shows the properties we gave in Theorems \ref{thm:contiguous} to \ref{thm:rho2}.

\begin{figure}
  \centering
  \includegraphics[width=3in]{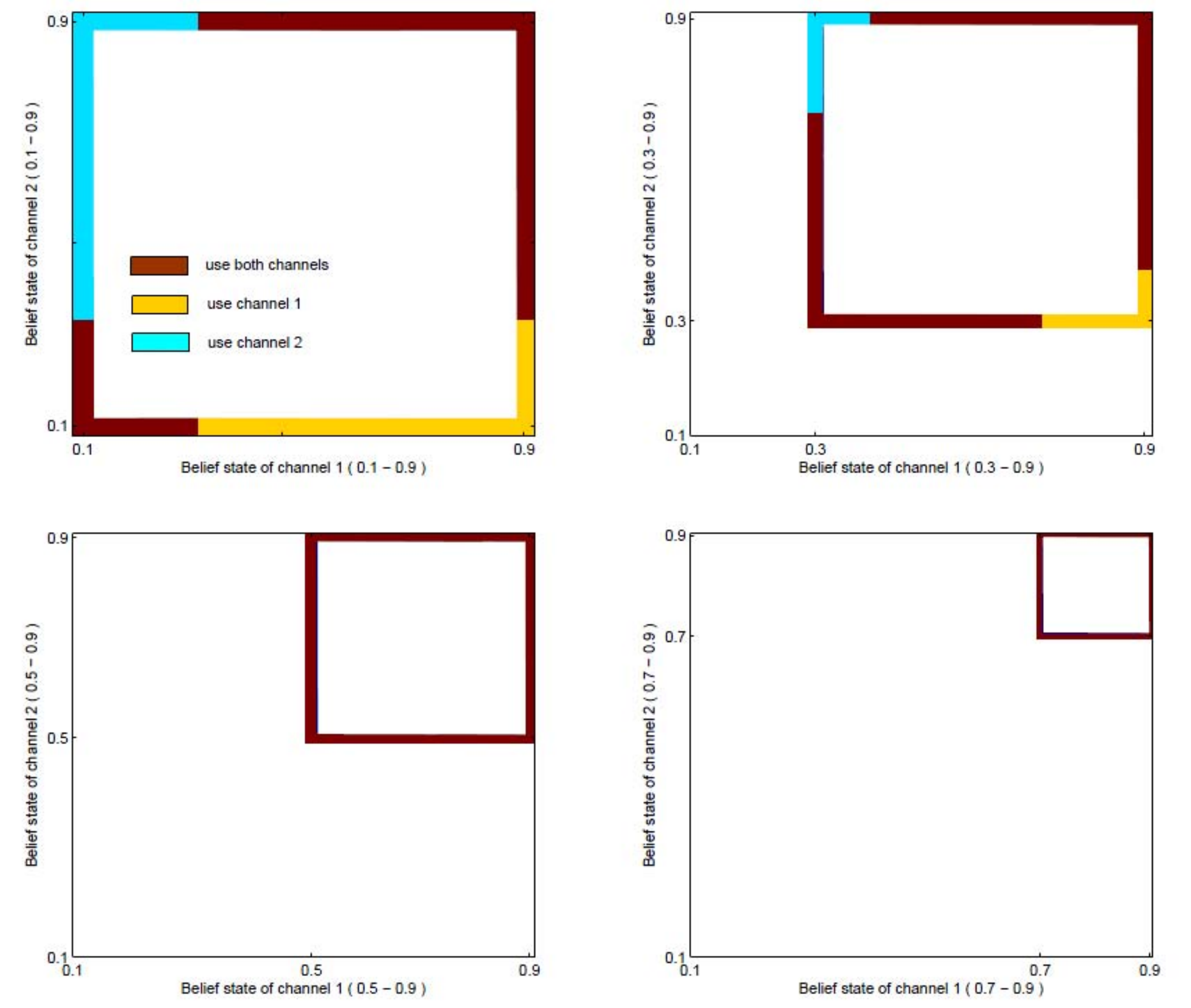}
  \caption{Optimal policy with increasing $\lambda_0$ ($R_l=2$, $R_h=3$) .}
  \label{fig:0109-0809}
\end{figure}

In order to observe the effect of parameters $\lambda_0, \lambda_1$, $R_l$ and $R_h$ on the structure of optimal policy, we have conducted simulation experiments for varying parameters. First, we increase $\lambda_0$ from 0.1 to 0.7, keeping the rest of the parameters the same as in the above experiment. Figure \ref{fig:0109-0809} shows the policy structure with different $\lambda_0$. We can observe in Figure \ref{fig:0109-0809} that when $\lambda_0$ increases from 0.1 to 0.3, the decision region of action $B_b$ occupies a bigger part of the belief space. Whilst when $\lambda_0$ is 0.5 or greater, the whole belief space falls in the decision region of action $B_b$, meaning that it is optimal to always use both channels in the set of this experiment when $\lambda_0 > 0.5$.

\begin{figure}
\centering
\includegraphics[width=3in]{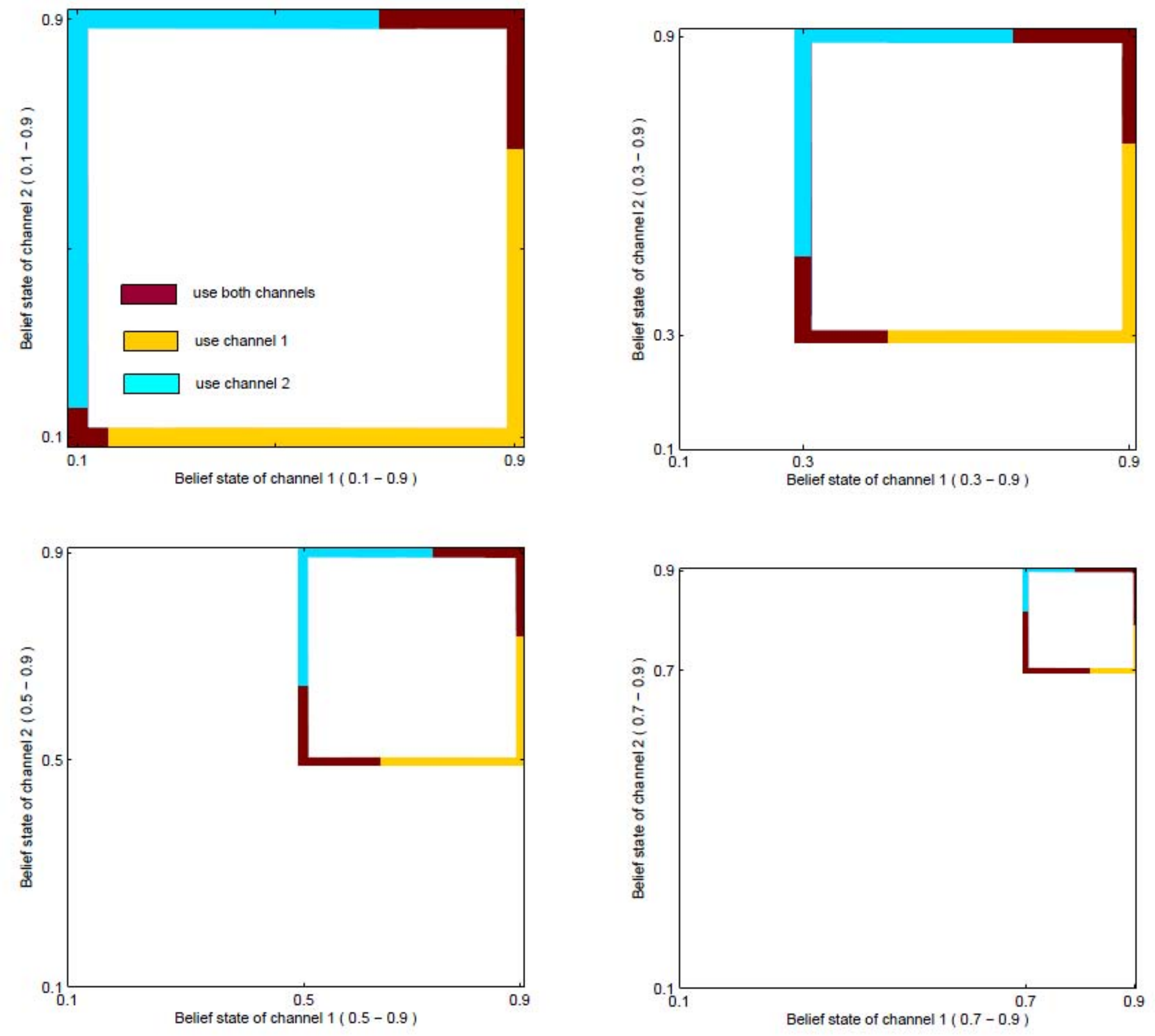}
\caption{Optimal policy with increasing $\lambda_0$ ($R_l=2, R_h=3.8$).}
\label{fig:01092038-08092038}
\end{figure}

Intuitively we believe the optimal policy is closely related to the immediate reward of the three actions. Therefore, in the next experiment, we set $R_h$ to 3.8 ($2= R_l < R_h=3.8 < 2R_l$) and repeat the above experiment, and the result is shown in Figure \ref{fig:01092038-08092038}. Compared to Figure \ref{fig:0109-0809}, the most obvious difference is that there is no zero-threshold policy structure in Figure \ref{fig:01092038-08092038}. This is because when the immediate reward of using one channel is greater (bigger $R_h$ ), it is more beneficial to use one channel than using both channels.

\begin{figure}
  \centering
   \includegraphics[width=3in]{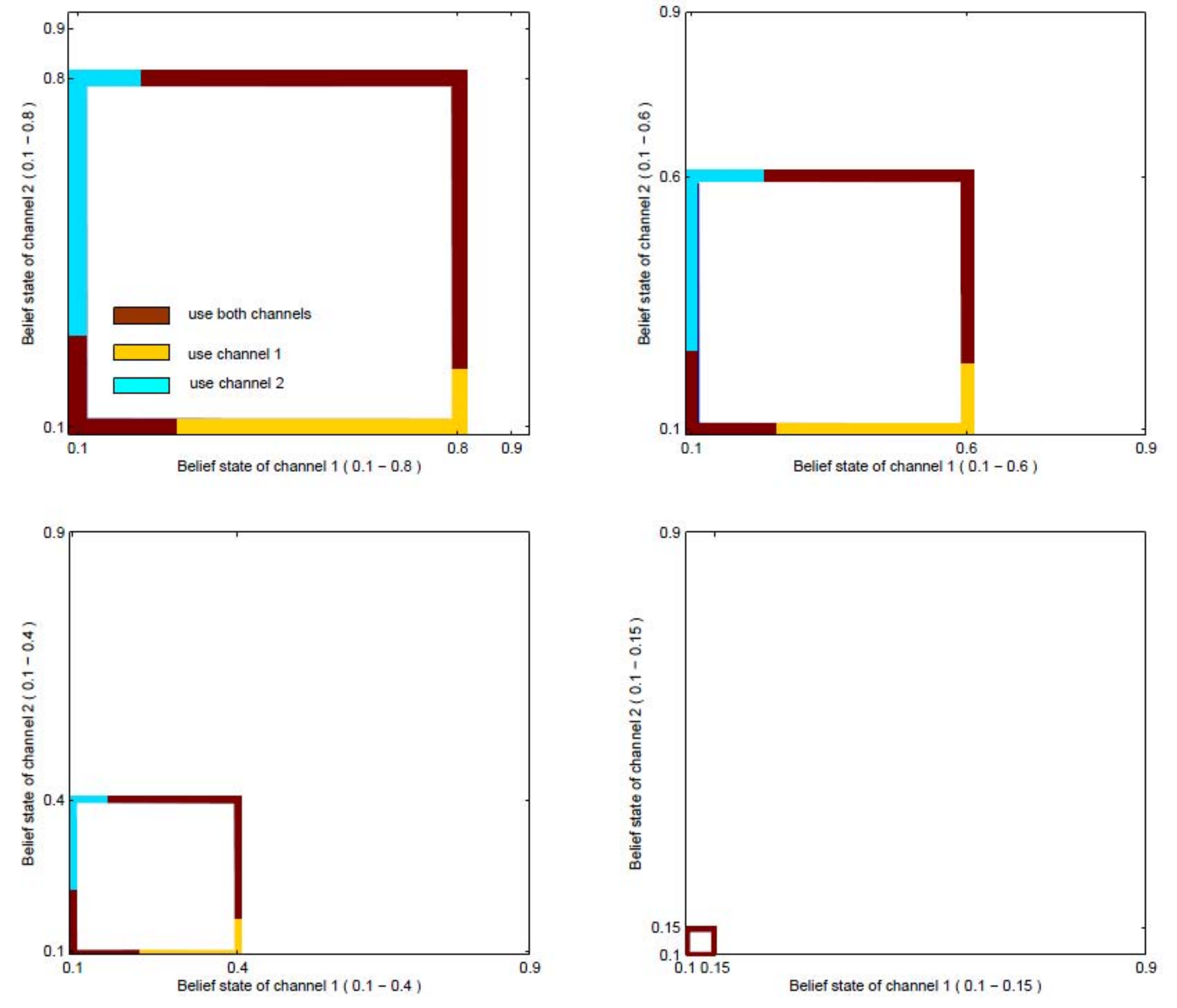}
  \caption{Optimal policy with decreasing $\lambda_1$ ($R_l=2, R_h=3.8$).}
  \label{fig:0109-0102}
\end{figure}

Figure \ref{fig:0109-0102} shows the structure of optimal policy when $\lambda_1$ decreases from 0.9 to 0.15. Other parameters in this experiment are: $\lambda_0=0.1, R_l=2,R_h=3, \beta =0.9$. As in Figure \ref{fig:0109-0809}, both two-threshold and zero threshold policies are observed in this experiment.

\begin{figure}
\centering
\includegraphics[width=1.7in]{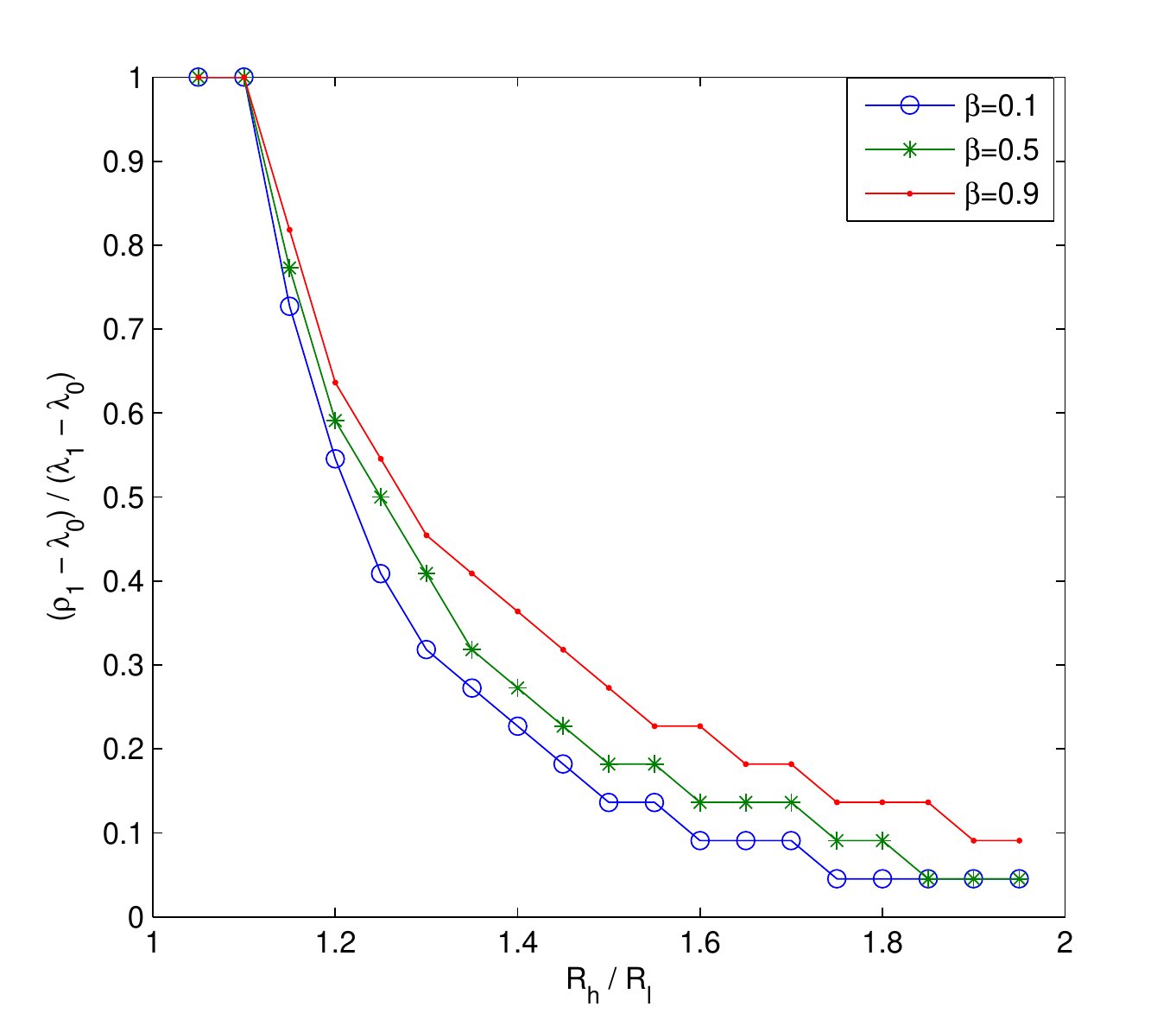}
\includegraphics[width=1.7in]{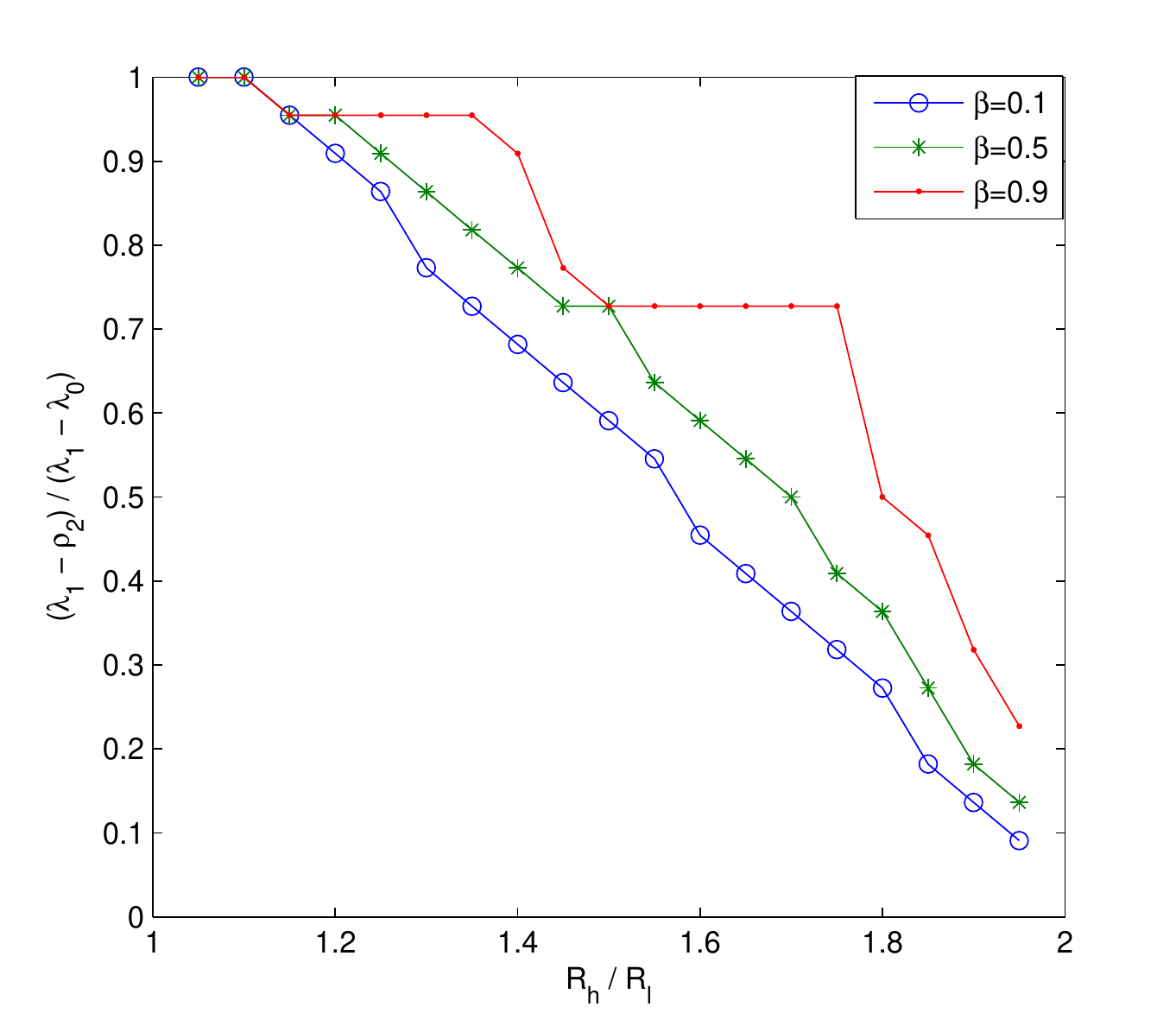}
\caption{Normalized threshold $\rho_1$ and $\rho_2$ with different $\beta$. ($\lambda_0=0.1,\lambda_1=0.9$)}
\label{fig:length1}
\end{figure}

From the above observation we understand that the thresholds are sensitive to the value of $R_l$ and $R_h$. Therefore, next we try to observe the relationship between the thresholds and the value of $R_l$ and $R_h$. Figure \ref{fig:length1} shows the normalized thresholds $\rho_1$ and $\rho_2$ versus the ratio of $R_h$ and $R_l$, with different discount factor $\beta$, when $\lambda_0=0.1, \lambda_1=0.9$. $\rho_1$ is normalized as $(\rho_1-\lambda_0)/(\lambda_1-\lambda_0)$, representing the relative length of $\Phi_{B_b}^{p_2=\lambda_0}$. Similarly, the normalized $\rho_2$, that is,  $(\lambda_1 -\rho_2)/(\lambda_1-\lambda_0)$, is the relative length of $\Phi_{B_b}^{p_2 = \lambda_1}$. It is clear to see that when the normalized threshold $\rho_1$ is 1 ($\rho_2$ is also 1), it corresponds to the zero threshold structure of the optimal policy. From Figure \ref{fig:length1} we can also observe that the structure of the optimal policy is affected by the value of the discount factor $\beta$.

\begin{figure}
\centering
\includegraphics[width=1.7in]{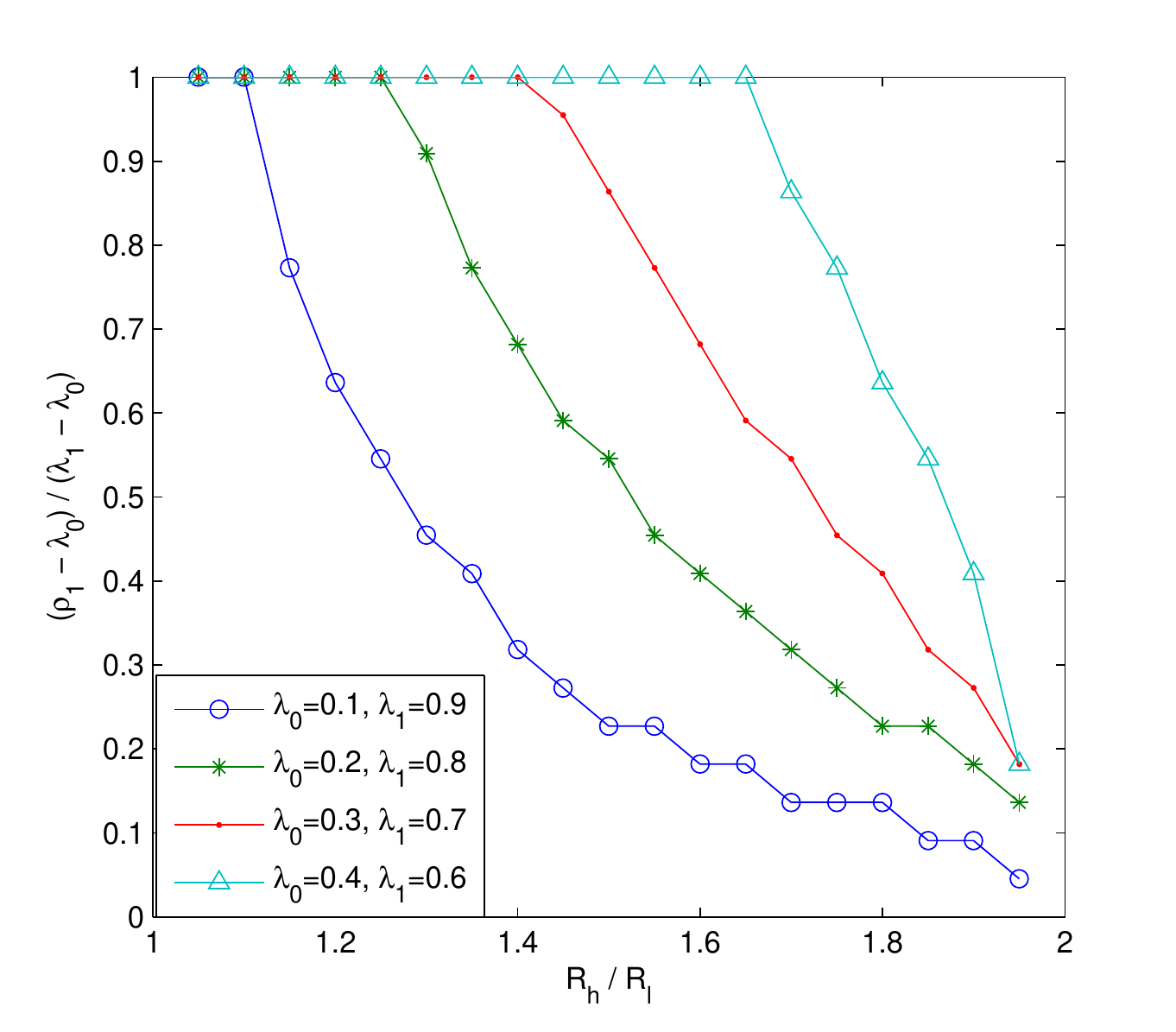}
\includegraphics[width=1.7in]{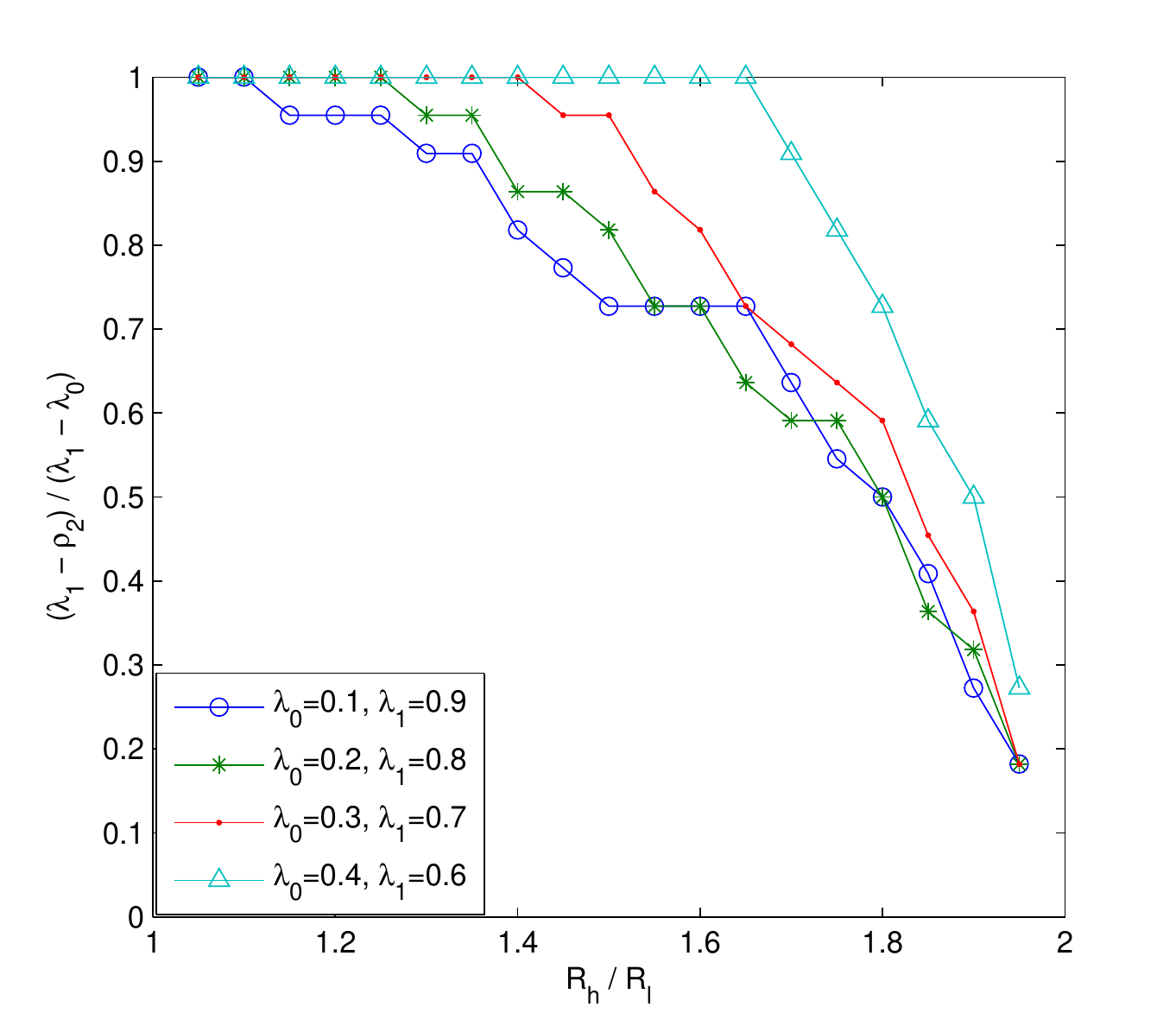}
\caption{Normalized threshold $\rho_1$ and $\rho_2$. ($\beta =0.8$)}
\label{fig:threshold4}
\end{figure}

Figure \ref{fig:threshold4} shows the normalized threshold $\rho_1$ and $\rho_2$ with different values of $\lambda_0$ and $\lambda_1$. Figure \ref{fig:threshold4} gives us a whole picture of the structure of the optimal policy for different $R_h/R_l$ and different size of the belief space. We can see that in all experiments with a wide range of parameters, no other policy structure than zero-threshold and two-threshold structure is observed. So we can conclude that with the help of linear-programming simulation, once the five parameters are known ($\lambda_0, \lambda_1, R_l, R_h, \beta$), the thresholds can be derived based on Figure \ref{fig:threshold4} and the exact optimal policy can be constructed.

\section{Conclusion}
 In this paper we have shown the structure of the optimal policy by theoretical analysis and simulation. Knowing that this problem has a 0 or 2 threshold structure reduces the problem of identifying optimal performance to finding the (only up to 2) threshold parameters. In settings where the underlying state transition matrices are unknown, this could be exploited by using a multiarmed bandit (MAB) formulation to find the best possible thresholds (similar to the ideas in the papers \cite{yanting2012} and \cite{nayyar2011}). Also, we would like to investigate the case of non-identical channels, and derive useful results for more than 2 channels, possibly in the form of computing the Whittle index~\cite{whittle1980}, if computing the optimal policy in general turns out to be intractable.

\section*{Acknowledgment}
This work is done when Junhua Tang is a visiting scholar at USC. The authors would like to thank Yi Gai and Yanting Wu for their helpful discussions. This work is partially supported by Natural Science Foundation of China under grant 61071081 and 60932003. This research was also sponsored in part by the U.S. Army Research Laboratory under the Network Science Collaborative Technology Alliance, Agreement Number W911NF-09-2-0053, and by the Okawa Foundation, under an Award to support research on ``Network Protocols that Learn".


\end{document}